\documentclass[twocolumn,groupeaddress,nofootinbib]{svjour3}
\smartqed  
\usepackage{graphicx}
\usepackage{mathptmx}      
\usepackage[english]{babel}
\usepackage{amssymb}
\usepackage{amsmath}
\usepackage{amscd}
\usepackage{eucal}
\usepackage{color}
\usepackage{bm}

\usepackage[bookmarks=true,colorlinks,linkcolor=red,urlcolor=blue,citecolor=blue]{hyperref}
\usepackage{dcolumn}

\usepackage{stmaryrd} 
\usepackage{algorithm}
\usepackage{algpseudocode}

\newcommand\be{\begin{equation}}
\newcommand\ee{\end{equation}}

\usepackage{ulem}

\journalname{Statistics and Computing}

\begin{document}

\title{A critical analysis of resampling strategies for the \\ regularized particle filter\thanks{pierre.carmier@centralesupelec.fr}}

\author{Pierre Carmier \and Olexiy Kyrgyzov \and Paul-Henry Courn\`ede}
\newcommand{\digiplante}{Biomathematics Team, MICS Lab, \\ CentraleSup\'elec, Universit\'e Paris-Saclay, France}
\institute{\digiplante}
\date{\today}

\maketitle

\begin{abstract}
We analyze the performance of different resampling strategies for the regularized particle filter regarding parameter estimation. We show in particular, building on analytical insight obtained in the linear Gaussian case, that resampling systematically can prevent the filtered density from converging towards the true posterior distribution. We discuss several means to overcome this limitation, including kernel bandwidth modulation, and provide evidence that the resulting particle filter clearly outperforms traditional bootstrap particle filters. Our results are supported by numerical simulations on a linear textbook example, the logistic map and a non-linear plant growth model. 
\end{abstract}

\section{Introduction}

Sequential Monte Carlo (SMC) methods \cite{DelMoral04,SMC} are powerful algorithms which allow performing Bayesian inference in dynamic systems. These methods, also known under the heading of particle filters in the context of dynamic systems, have become extremely popular way beyond the borders of the statistics community and are nowadays widely used by engineers and scientists from other fields faced with non-trivial inference tasks which naturally abound with the emergence of data science. The purpose of this paper is to provide a critical analysis of the performance of regularized particle filter (RPF) algorithms \cite{LeGland98,Musso12} which represent an interesting variation of traditional particle filters. They rely on kernel-based \cite{Parzen62} posterior density estimation at each filtering step, with particle resampling performed from the reconstructed density. In many situations, this method was proven to help prevent sample impoverishment \cite{SMC} and thereby enhance the robustness of particle filtering. Our main achievement will be to demonstrate in this paper that the approach which consists in resampling systematically can lead to seriously over-estimate the posterior mean squared error, motivating us to exhibit alternative resampling strategies which allow fully resolving this issue.

Before going any further, let us introduce the general framework we shall be working with in the rest of this paper. To begin with, we shall assume that the processes from which we seek to infer some information belong to the family of hidden Markov models (HMM) \cite{HMM}, which describe stochastic processes characterizing the evolution in discrete time $n$ of two sets of random variables $\{{\bf x}_n\}_{n\geq 0}$ and $\{{\bf y}_n\}_{n\geq 1}$, respectively defined on some space ${\cal X} \subseteq \mathbb{R}^{d_x}$ and ${\cal Y} \subseteq \mathbb{R}^{d_y}$. The former are hidden random variables referred to as states, while the latter are noisy observations of the former. Starting from a prior probability measure $\pi_0(d{\bf x}_0)$ on the initial state, hidden variables follow a Markovian evolution characterized by a sequence of Markov kernels $f_n(d{\bf x}_n | {\bf x}_{n-1})$ also named transition functions. Observations are considered conditionally independent and described by a conditional probability density function $g_n({\bf y}_n | {\bf x}_n)$, the so-called measurement function, with respect to the Lebesgue measure $\lambda$. Parameters $\{\Theta\}$, with a compact support set, can easily be accounted for in that framework by augmenting the size of the state space, as can the presence of environmental variables $\{{\bf u}_n\}_{n\geq 0}$ defined on some space ${\cal U} \subseteq \mathbb{R}^{d_u}$.

A problem of broad interest in this context is to determine the posterior probability measure $\pi_k(d{\bf x}_k | {\bf y}_{1:n})$ given a sequence of observations ${\bf y}_{1:n}=\{{\bf y}_1,...,{\bf y}_n\}$. If $k<n$, this is called a smoothing problem, while if $k>n$ this is a prediction problem. In the following we shall be interested in the situation where $k=n$, which is called the filtering problem. Thus, given any Borel set ${\cal A} \subset {\cal X}$, $\pi_n({\cal A}) = \int_{\cal A} \pi_n(d{\bf x}_n | {\bf y}_{1:n})$ shall be the posterior probability of event ${\cal A}$. More generally, for any test function $\phi$ defined on ${\cal X}$, applying Bayes rule recursively allows expressing the integral of $\phi$ with respect to the posterior probability measure as a function of past measurement and transition functions \cite{HMM}
\be
\begin{split}
& \int_{\cal X} \phi({\bf x}_n)\pi_n(d{\bf x}_n|{\bf y}_{1:n})
\\
& = \frac{1}{Z_n} \int_{{\cal X}^{n+1}} \phi({\bf x}_n) \pi_0(d{\bf x}_0) \prod_{j=1}^n g_j({\bf y}_j | {\bf x}_j) f_j(d{\bf x}_j | {\bf x}_{j-1}) \; ,
\end{split}
\ee
introducing the marginal likelihood
\be
Z_n = \int_{{\cal X}^{n+1}} \pi_0(d{\bf x}_0) \prod_{j=1}^n g_j({\bf y}_j | {\bf x}_j) f_j(d{\bf x}_j | {\bf x}_{j-1}) \; .
\ee
The purpose of SMC methods \cite{Liu98,DelMoral06,SMC} is to provide an estimation of such probability measures in the general case of nonlinear state-space models, where analytical attempts are essentially hopeless. The idea of such methods is to represent the probability distributions by particles, \textit{i.e.} independent identically distributed samples of the distributions, and to propagate and update them using the HMM equations, hence the equivalent nomenclature of particle filter methods. The estimation of parameters in this context has received acute attention from the community these last years (see \cite{Kantas15} for a review). One of the main challenges is to circumvent the celebrated particle degeneracy problem \cite{Gordon93} and to ensure that parameter space is sufficiently explored to enable accurate estimation. Several methods have been proposed to address this challenge, among which iterated filtering \cite{Ionides11} and particle Markov chain Monte Carlo \cite{Andrieu10} in the specific context of joint parameter / hidden state estimation. They constitute batch or offline methods, however, in that they are unable to take into account new observations dynamically, as opposed to online algorithms such as SMC$^2$ \cite{Chopin13,Jacob15}. 

A different idea, introduced some time ago in \cite{West93} and discussed in more detail in \cite{Liu10,Musso12}, is to allow freedom in parameter space by regularizing the posterior density using a mixture. In this paper, we focus on the RPF introduced in Ref.~\cite{LeGland98} and further developed and discussed in Refs~\cite{Musso12,Campillo09,ChenPhD,Gauchi13}. We shall highlight some shortcomings of this approach when resampling is performed systematically (as routinely done), which can be made explicit by computing the posterior mean squared error. In particular, building upon some well-known results obtained for linear Gaussian state-space models, we prove that the mean squared error is lower bounded by a positive constant when the number of particles is finite, independently of the number of observations. We shall provide a detailed understanding of this phenomenon on a very simple linear textbook example, offering evidence that it is related to the particle filter losing memory of past observations. We will then show that long-ranged memory can be restored by appropriately modulating the kernel bandwidth and that the efficiency of this suitably modified particle filter clearly outperforms that of the original bootstrap particle filter \cite{Gordon93}. The rest of the paper is organized as follows. In section \ref{sec:PF}, we recall basic principles of SMC methods and introduce the bootstrap particle filter and the RPF, before illustrating what can potentially go wrong when resampling is performed on a systematic basis. We then provide analytical insight regarding the effect of different resampling strategies on the behavior of the RPF in section \ref{sec:KF}, using a very simple linear textbook example for illustrative purposes. That our results qualitatively hold for more general, non-linear, models is discussed in section \ref{sec:NL}. Conclusions are drawn in section \ref{sec:Concl}, while a series of technical details are gathered in appendices.

\section{Sequential Monte Carlo methods}
\label{sec:PF}

Let us begin by recalling the basic principles of particle filter methods. Starting from the initial prior probability measure from which $N$ particles are sampled, ${\bf x}_0^{(i)} \sim \pi_0(d{\bf x}_0)$, particles are propagated until the next time step using a Markov kernel generally chosen \cite{Gordon93,Kong94} as ${\bf x}_{1|0}^{(i)} \sim f_1(d{\bf x}_1|{\bf x}_0^{(i)})$ in nonlinear state space models for which $f_1(d{\bf x}_1|{\bf x}_0^{(i)}, {\bf y}_1)$ is not readily available \cite{Liu98}. The resulting empirical estimator for the predicted probability measure can be expressed as 
\be
\pi_{1|0}^N(d{\bf x}) = \frac{1}{N}\sum_{i=1}^N \delta_{{\bf x}_{1|0}^{(i)}}(d{\bf x})
\ee
with $\delta$ the Dirac measure. The estimator for the updated posterior probability measure is obtained by weighting each particle according to the observation ${\bf y}_1$:
\be
\label{eq:multinomial}
\pi_1^N(d{\bf x}_1 | {\bf y}_1) = \sum_{i=1}^N w_1^{(i)} \delta_{{\bf x}_{1|0}^{(i)}}(d{\bf x}_1)
\ee
with weights
\be
w_1^{(i)} = \frac{g_1({\bf y}_1 | {\bf x}_{1|0}^{(i)})}{\sum_{j=1}^Ng_1({\bf y}_1 | {\bf x}_{1|0}^{(j)})}
\ee
suitably normalized to ensure that total weight is conserved. Given any bounded function $\phi$ defined on ${\cal X}$, its integral with respect to the empirical posterior probability measure is then given by 
\be
\int_{\cal X} \phi({\bf x}_1) \pi_1^N(d{\bf x}_1|{\bf y}_1) = \sum_{i=1}^N w_1^{(i)} \phi({\bf x}_{1|0}^{(i)}) \; ,
\ee
and the estimator can be shown to converge asymptotically, for an infinite number of particles, towards the true posterior probability measure given by Bayes rule \cite{Chopin04}:
\be
\lim_{N\to\infty} \int_{\cal X} \phi({\bf x}_1) \pi_1^N(d{\bf x}_1|{\bf y}_1) = \int_{\cal X} \phi({\bf x}_1) \pi_1(d{\bf x}_1|{\bf y}_1) \; .
\ee
Practically speaking, however, one always works with a finite number of particles and life is not as simple. A well-known problem which can occur with the above estimation procedure is the so-called weight degeneracy issue \cite{SMC}. If the likelihood is relatively peaked around the observation value, or in other words if the observation noise is rather small, particle weights will differ by orders of magnitude depending on the distance between the particle and the observation. Actually, this will happen almost regardless of the magnitude of the observation noise for a sufficient number of observations, since iterating the above procedure yields $w_n^{(i)} \propto \prod_{j=1}^n g_j({\bf y}_j|{\bf x}_{j|j-1}^{(i)})$, using the recursive formulation ${\bf x}_{j|j-1}^{(i)} \sim f_j(d{\bf x}_j|{\bf x}_{j-1}^{(i)})$. This implies that, eventually, most of the weight will be carried by a small number of particles, with a dramatic effect on the precision of the resulting estimation as can be seen by looking at the unbiased estimator of the marginal likelihood at step $n$:
\be
Z_n^N = \frac{1}{N}\sum_{i=1}^N \prod_{j=1}^n g_j({\bf y}_j|{\bf x}_{j|j-1}^{(i)}) \; .
\ee
Its variance can be estimated using the standard empirical estimator 
\be
\frac{1}{N^2}\sum_{i=1}^N \left(\prod_{j=1}^n g_j({\bf y}_j|{\bf x}_{j|j-1}^{(i)})-Z_n^N \right)^2 = \frac{(Z_n^N)^2}{\text{ESS}_n}\left(1 - \frac{\text{ESS}_n}{N}\right) \; ,
\ee
where we introduced the effective sample size 
\be
\label{eq:ESS}
\text{ESS}_n = \frac{1}{\sum_{i=1}^N (w_n^{(i)})^2}
\ee
which is a measure of the effective number of particles \cite{Kong94,Carpenter99}: if all particles carry the same weight, $\text{ESS}_n=N$, while if a single particle carries all the weight, $\text{ESS}_n=1$. This result therefore illustrates that the accuracy of the estimation by the empirical measure relies on a relatively homogeneneous distribution of particle weights. 

A natural strategy to counter the effect of particle degeneracy is to 'unweight' particles by resampling from the empirical posterior probability measure, ${\bf x}_n^{(i)} \sim \pi_n^N(d{\bf x}_n|{\bf y}_{1:n})$, such that the new approximation to the posterior is
\be
\tilde{\pi}_n^N(d{\bf x}_n | {\bf y}_{1:n}) = \frac{1}{N}\sum_{i=1}^N \delta_{{\bf x}_n^{(i)}}(d{\bf x}_n) \; .
\ee 
This is the idea behind the bootstrap particle filter \cite{Gordon93}, also known as the sequential importance resampling (SIR) algorithm. Its main drawback is that resampling according to a probability measure such as Eq.~(\ref{eq:multinomial}) essentially duplicates particles (particles are selected following a multinomial law), leading to so-called sample impoverishment \cite{SMC}.

\subsection{Regularized particle filter}

In Ref.~\cite{LeGland98}, LeGland, Musso and Oujdane proposed to address these shortcomings by regularizing the probability measure using a Parzen-Rosenblatt \cite{Parzen62} kernel density estimator. The idea is to replace the Dirac measure by a measure able to smoothly interpolate between all weighted particles. In other words, observing that the coarse-grained density on Borel set ${\cal A}_\epsilon\subset{\cal X}$
\be
\frac{1}{\lambda({\cal A}_\epsilon)}\pi_n^N({\cal A}_\epsilon) = \frac{1}{\lambda({\cal A}_\epsilon)}\sum_{i=1}^N w_n^{(i)} \delta_{{\bf x}_{n|n-1}^{(i)}}({\cal A}_\epsilon) \; ,
\ee
where $\lambda({\cal A}_\epsilon)=\int_{{\cal A}_\epsilon} \lambda(d{\bf x}_n)$, the regularization procedure consists in replacing $\delta_{{\bf x}}({\cal A}_\epsilon)/\lambda({\cal A}_\epsilon)$ as $\lambda({\cal A}_\epsilon) \to 0$ by a kernel density function with respect to the Lebesgue measure \cite{Tsybakov}, thanks to which the regularized estimator for the posterior probability density function reads
\be
\label{eq:RPF}
p_n^N({\bf x}_n | {\bf y}_{1:n}) = \sum_{i=1}^N w_n^{(i)} {\cal K}_{h_n}({\bf x}_n - {\bf x}_{n|n-1}^{(i)}) \; ,
\ee
with 
\begin{itemize}
\item[$\bullet$] $w_n^{(i)}=g_n({\bf y}_n | {\bf x}_{n|n-1}^{(i)}) / \sum_{j=1}^N g_n({\bf y}_n | {\bf x}_{n|n-1}^{(j)})$
\item[$\bullet$] ${\bf x}_{n|n-1}^{(i)} \sim f_n(d{\bf x}_n | {\bf x}_{n-1}^{(i)})$
\item[$\bullet$] ${\bf x}_{n-1}^{(i)} \sim p_{n-1}^N({\bf x}_{n-1} | {\bf y}_{1:n-1})$ \; .
\end{itemize}
In the following, we shall consider a Gaussian kernel ${\cal K}_{h_n} = {\cal N}(0, h_n^2)$: Eq.~(\ref{eq:RPF}) then effectively describes a Gaussian mixture and resampling according to this distribution now consists in a two-step procedure, with the particle selection (performed using the systematic resampling algorithm \cite{Hol06}) complemented by a perturbation controlled by the kernel bandwidth $h_n$. Each particle generated in the process is thus unique. Note that, if the likelihood function $g_n$ is unknown, it can itself be emulated by a kernel density estimator \cite{Campillo09,ChenPhD}.

There is a whole literature on how the kernel bandwidth $h_n$ can be chosen. For now, we will follow Ref.~\cite{Musso12} and pick $h_n^2=\alpha_h\Sigma_n$, with
\be
\label{eq:bandwidth}
\alpha_h = \left(\frac{4}{N(d_x+2)}\right)^{\frac{2}{d_x+4}} \; .
\ee
This choice minimizes the mean integrated squared error $\mathbb{E}\left[ ||p_n^N - p_n||_{L^2}^2 \right]$, provided the true probability density function $p_n$ is Gaussian (see Appendix \ref{app:band} for the proof in one dimension with uniform weights). In practice, the true covariance matrix $\Sigma_n$ is unknown and is therefore replaced by an estimator $\Sigma_n^N$ based on the weighted samples $({\bf x}_{n|n-1}^{(i)}, w_n^{(i)})$: 
\be
\Sigma_n^N=\frac{N}{N-1}\sum_{i=1}^N w_n^{(i)}({\bf x}_{n|n-1}^{(i)}-{\bf \mu}_n^N)^2
\ee
with ${\bf \mu}_n^N=\sum_{i=1}^N w_n^{(i)}{\bf x}_{n|n-1}^{(i)}$.
As an important corollary, because $\lim_{N\to\infty}h_n=0$, we retrieve asymptotically that, just as for the empirical probability measure, the regularized estimator converges towards the true posterior density: 
\be
\mathbb{E}\left[ ||p_n^N - p_n||_{L^2} \right]\underset{N\to\infty}{=}\mathcal{O}\left(N^{-\frac{2}{d_x+4}}\right) \; .
\ee

\subsection{Problem statement}
\label{sec:rr=1}

While the general strategy underlying the RPF seems very principled, we proceed to highlight a problem which can be encountered when resampling is performed on a systematic basis, {\it i.e.} at every step, as generally done. Recall that, given a prior probability measure $\pi_0(d{\bf x}_0)$ and a series of observations $\{{\bf y}_{1:n}\}$, we seek to determine the true posterior distribution $\hat{{\bf x}}_n \sim \pi_n(d{\bf x}_n|{\bf y}_{1:n})$. In particular, we wish to estimate how the RPF estimator $\hat{{\bf x}}_n^N \sim p_n^N({\bf x}_n|{\bf y}_{1:n})$ fares with respect to $\hat{{\bf x}}_n$. In what follows, we shall provide evidence that, for a finite number of particles, the RPF estimator features a systematic deviation from the true posterior density which increases with the number of observations. This deviation can be made explicit by comparing the root mean squared error of the RPF with that of the true posterior density:
\be
\begin{split}
\text{RMSE}(\hat{{\bf x}}_n, {\bf x}_n) & = \mathbb{E}\left[||\hat{{\bf x}}_n - {\bf x}_n||^2\right]^{1/2} 
\\
& = \sqrt{||{\bf \mu}_n-{\bf x}_n||^2+\text{Tr}[\Sigma_n]} \; ,
\end{split}
\ee
where $||.||$ indicates the usual Euclidean norm, and introducing the mean value ${\bf \mu}_n := \mathbb{E}[\hat{{\bf x}}_n | {\bf y}_{1:n}]$ and the covariance matrix $\Sigma_n := \mathbb{E}[(\hat{{\bf x}}_n - {\bf \mu}_n)(\hat{{\bf x}}_n  - {\bf \mu}_n)^T]$. For illustrative purposes, consider the very simple case of a one-dimensional stationary state $x_{n+1}=x_n$ with time-independent additive observation noise $\eta \sim {\cal N}(0, R)$: $y_n = x_n + \eta$. This is of course an elementary textbook setting, which is in fact a special case of the Kalman filter framework which shall be reviewed in the following section. The advantage of this setup is to provide us with the exact result with respect to which particle filter expectations can be benchmarked. In the left frame of Fig.~\ref{fig:1}, we plot the analytical solution and compare it with the outcome of the RPF when resampling is performed at each step.
\begin{figure*}
\begin{center}
  \includegraphics[angle=0,width=0.45\linewidth]{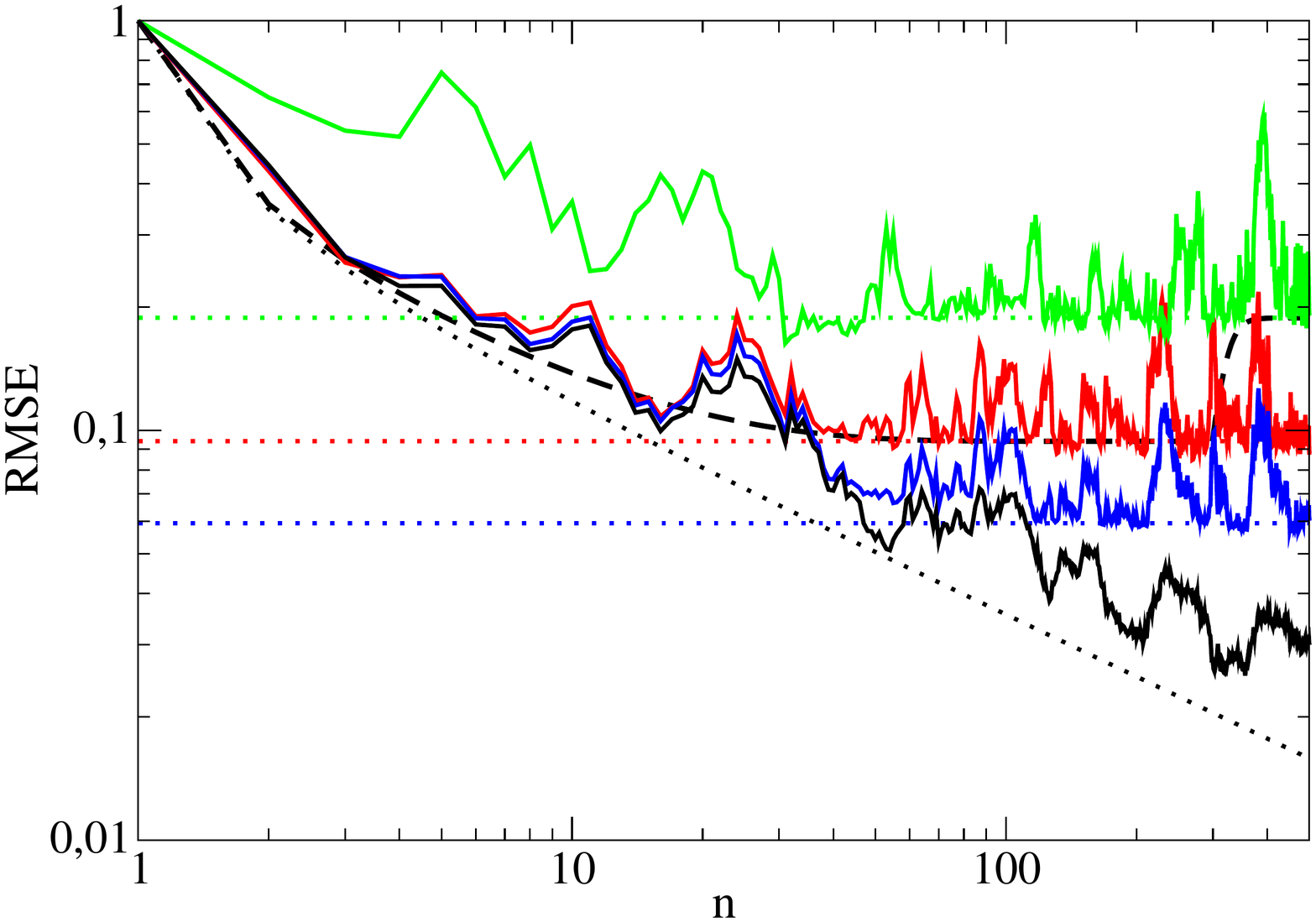}
  \includegraphics[angle=0,width=0.45\linewidth]{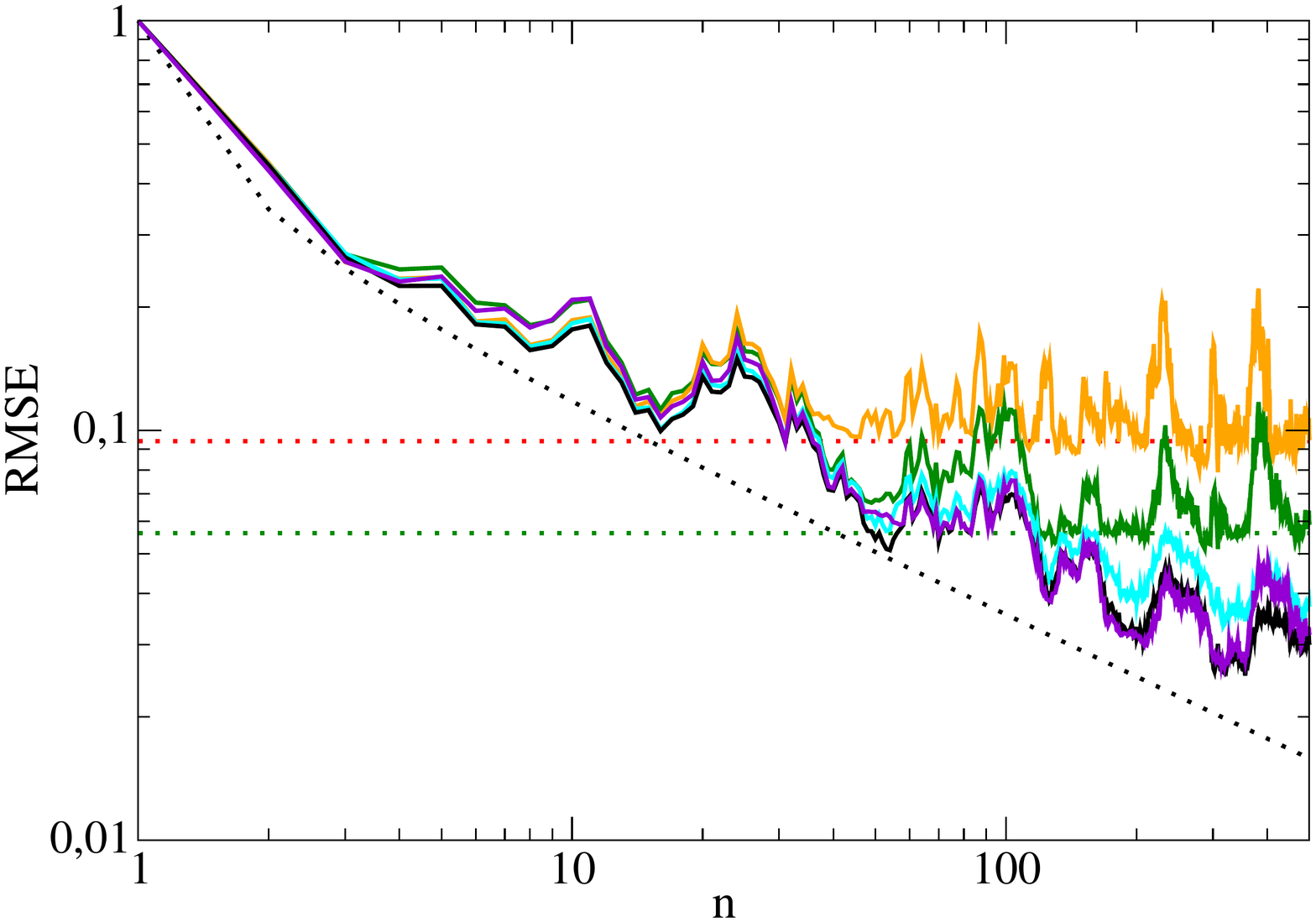}
\caption{Left frame: Normalized root mean squared error (RMSE) of the estimated state as a function of the number of observations in logarithmic scale for $N=1000$ particles and normalized observation noise variance $R/\Sigma_0=0.25$. The black solid line is the expectation based on the Kalman filter, while the red solid line is the outcome of the RPF with resampling at each step. The latter clearly shows a saturation effect at a value given by the red dotted line, which we compute analytically in section \ref{sec:KF}. Blue lines are for $N=10000$ particles and green lines for $R/\Sigma_0=1$. The black dashed line is obtained by taking into account the bandwidth perturbation inside the Kalman filter, with a quench in normalized observation noise variance from 0.25 to 1 at $n=300$. Other parameters: prior mean value $\mu_0=x_0+\sqrt{\Sigma_0}$. Right frame: Likewise, for different bandwidth selection criteria: Silverman's rule of thumb $h_n^2=(4/N)^{2/3}\Sigma_n$ (dark green), direct plugin method by Sheather and Jones \cite{Sheather04} (orange), West shrinkage \cite{West93} (cyan) and our proposal of kernel bandwidth modulation $\alpha_n=\alpha_h/(1+n\alpha_h)$ (violet).}
\label{fig:1}
\end{center}
\end{figure*}
As can be clearly seen, while the analytical expectation steadily decreases with each new observation, the RMSE computed using the RPF rather quickly saturates and remains lower bounded by a positive constant no matter how many observations are considered. The explanation for this behavior essentially stems from the fact that resampling from the regularized estimator, while allowing to explore particle space, also acts as a perturbation which increases the variance of the estimator. This perturbation becomes detrimental for the estimation process in the long run, as is manifest in the left frame of Fig.~\ref{fig:1}. While this problem was already acknowledged in a slightly different context in the early days of particle filters \cite{West93}, some of its implications may not have been fully grasped at the time. In particular, we claim that the over-estimation of the variance prevents the estimator from remaining asymptotically efficient as the number of observations increases, in a sense which shall be defined shortly. Mathematical arguments supporting this claim for linear state-space models will be provided in the following section, while numerical simulations indicating that our results extend to non-linear models will be presented in section \ref{sec:NL}.

\section{Theoretical results in the linear Gaussian case}
\label{sec:KF}

We now address the origin of the saturation effect in the estimated RMSE in more detail and claim that it can be understood using results which we obtain in the setting of linear Gaussian state-space models. Before providing mathematical evidence supporting our claim, we begin by recalling standard results for the derivation of the filtered probability density function in the linear Gaussian case, which is the framework of the famous Kalman filter.

\subsection{Kalman filter framework}

The framework of the Kalman filter \cite{Kalman60,Kalman61} is obtained by considering the special case where transition and measurement functions are linear, and assuming that noise and prior probability density functions are Gaussian distributed with respect to the Lebesgue measure. The HMM equations then read
\be
\label{eq:linHMM}
\left|
\begin{array}{l}
{\bf x}_n = A_n {\bf x}_{n-1} + \epsilon_n
\vspace{0.1cm}
\\
{\bf y}_n = B_n {\bf x}_n + \eta_n
\end{array}
\right.
\ee
where $\{A_n\}_{n\geq 1}$ and $\{B_n\}_{n\geq 1}$ are referred to as transition and measurement matrices, and with $\epsilon_n \sim {\cal N}(0, Q_n)$ and $\eta_n \sim {\cal N}(0, R_n)$. In this situation, given a prior $\hat{{\bf x}}_0 \sim {\cal N}(\mu_0, \Sigma_0)$, the filtering problem can be solved exactly: analytical expressions can be obtained for the posterior $\hat{{\bf x}}_n \sim {\cal N}(\mu_n, \Sigma_n)$ which is fully characterized by its mean value ${\bf \mu}_n$ and (symmetric positive definite) covariance matrix $\Sigma_n$. Their expressions are provided by the recursive formulae
\be
\label{eq:Kalman}
\left|
\begin{array}{l}
\Sigma_n^{-1} = \Sigma_{n|n-1}^{-1} + B_n^TR_n^{-1}B_n
\vspace{0.1cm}
\\
\Sigma_n^{-1}{\bf \mu}_n = \Sigma_{n|n-1}^{-1}{\bf \mu}_{n|n-1} + B_n^TR_n^{-1}{\bf y}_n
\end{array}
\; ,
\right.
\ee
with ${\bf \mu}_{n|n-1} = A_n {\bf \mu}_{n-1}$ and $\Sigma_{n|n-1} = A_n \Sigma_{n-1} A_n^T + Q_n$ (see Appendix \ref{app:KF} for a derivation). A simple way of obtaining these results is to make use of the Bayesian conjugacy property of Gaussian distributions: in other words, given a Gaussian prior and a Gaussian likelihood, the posterior is also Gaussian. Applying Bayes rule to the predicted and posterior probability measures, respectively as
\be
\begin{split}
\pi_{n|n-1}(d{\bf x}_n | {\bf y}_{1:n-1})  = \int_{\cal X} & f_n(d{\bf x}_n | {\bf x}_{n-1}) 
\\
& \times \pi_{n-1}(d{\bf x}_{n-1} | {\bf y}_{1:n-1}) \; ,
\end{split}
\ee 
and
\be
\pi_n(d{\bf x}_n | {\bf y}_{1:n}) = \frac{g_n({\bf y}_n | {\bf x}_n) \pi_{n|n-1}(d{\bf x}_n | {\bf y}_{1:n-1})}{\int_{\cal X} g_n({\bf y}_n | {\bf x}_n) \pi_{n|n-1}(d{\bf x}_n | {\bf y}_{1:n-1})} \; ,
\ee
thus allows for a straightforward derivation of Eq.~(\ref{eq:Kalman}). We next focus on the special case where all states are assumed stationary. This corresponds to taking $A_k=\mathbb{I}$ and $Q_k=0$ for $k\geq 1$ in Eq.~(\ref{eq:linHMM}).
If we further fix the measurement covariance matrices to unity $B_k=\mathbb{I}$ and assume that observation noise covariance matrices are time-independent, one can show using the matrix inversion lemma that 
\be
\label{eq:KalmanParam1}
\left|
\begin{array}{l}
\Sigma_n = R(R+n\Sigma_0)^{-1}\Sigma_0
\vspace{0.1cm}
\\
{\bf \mu}_n = \Sigma_n\Sigma_0^{-1} {\bf \mu}_0 + (\mathbb{I} - \Sigma_n\Sigma_0^{-1}) \langle {\bf y} \rangle_n
\end{array}
\right.
\ee
with $\langle {\bf y} \rangle_n := \sum_{j=1}^n{\bf y}_j/n$ the maximum likelihood estimator. This situation corresponds to a multi-dimensional generalization of the stationary model we considered earlier. Note that $\Sigma_n$ reaches the Cramer-Rao bound $R/n$ as $n\to \infty$, while $\lim_{n\to\infty}{\bf \mu}_n = {\bf x}_0$: we thus retrieve that, in accordance with the Bernstein-von Mises theorem, the Bayes estimator $\hat{{\bf x}}_n$ is asymptotically efficient and coincides in that limit with the maximum likelihood estimator, as the dependence on the choice of prior vanishes. The corresponding $\text{RMSE} = \mathcal{O}(n^{-1/2})$ also vanishes asymptotically. This behavior might seem like a natural byproduct in the limit of a large number of observations, since the latter might be expected to eventually provide sufficient information for an arbitrarily accurate estimation of the true state. There are, however, elementary counterexamples to this fact which can be very easily proven using the Kalman filter equations we have derived. For example, if we restore a simple time-independent transition matrix $A=a \mathbb{I}$ to the above estimation paradigm, one can easily show that the asymptotic behavior of $\Sigma_n$ depends on the magnitude of $a$: while if $a<1$, the covariance matrix decays exponentially fast, if $a>1$ however, then the covariance matrix remains lower bounded as $\Sigma_n \geq (1-a^{-2})R$. Note that if $a=1$, we recover the algebraic dependence as stated in Eq.~(\ref{eq:KalmanParam1}). Another enlightening case is when state noise $Q\neq0$ is included. Given time-independent state and observation noises, one can show that $\Sigma_n$ necessarily remains lower bounded by a positive matrix \cite{HMM}. 

\subsection{Regularized RMSE asymptotics}

Let us now investigate how the above results can be transposed in the framework of the RPF. For a Gaussian regularizing kernel, Bayesian conjugacy ensures that the posterior distribution Eq.~(\ref{eq:RPF}) is asymptotically Gaussian:
\be
\label{eq:RPFMC}
\begin{split}
& p_n^N({\bf x}_n | {\bf y}_{1:n}) = \frac{\frac{1}{N} \sum_{i=1}^N g_n({\bf y}_n | {\bf x}_{n|n-1}^{(i)}) {\cal K}_{h_n}({\bf x}_n - {\bf x}_{n|n-1}^{(i)})}{\frac{1}{N} \sum_{j=1}^N g_n({\bf y}_n | {\bf x}_{n|n-1}^{(j)})}
\\
& \underset{N\to\infty}{\approx} \frac{\int_{\cal X} \lambda(d{\bf x}'_n) p_{n|n-1}^N({\bf x}'_n | {\bf y}_{1:n-1}) g_n({\bf y}_n | {\bf x}'_n) {\cal K}_{h_n}({\bf x}_n - {\bf x}'_n)}{\int_{\cal X} \lambda(d{\bf x}'_n) p_{n|n-1}^N({\bf x}'_n | {\bf y}_{1:n-1}) g_n({\bf y}_n | {\bf x}'_n)} \; ,
\end{split}
\ee
with mean 
\be
\label{eq:RPFMean}
\left((\Sigma_{n|n-1}^N)^{-1}+B_n^TR_n^{-1}B_n\right){\bf \mu}_n^N = (\Sigma_{n|n-1}^N)^{-1}{\bf \mu}_{n|n-1}^N + B_n^TR_n^{-1}{\bf y}_n
\ee
and covariance matrix
\be
\label{eq:RPFSigma}
\begin{split}
\Sigma_n^N & = h_n^2 + \left((\Sigma_{n|n-1}^N)^{-1}+B_n^TR_n^{-1}B_n\right)^{-1}
\\
& = (1+\alpha_h)\left((\Sigma_{n|n-1}^N)^{-1}+B_n^TR_n^{-1}B_n\right)^{-1} \; , 
\end{split}
\ee
with ${\bf \mu}_{n|n-1}^N = A_n {\bf \mu}_{n-1}^N$ and $\Sigma_{n|n-1}^N = A_n \Sigma_{n-1}^N A_n^T + Q_n$, and where $\alpha_h$ is defined in Eq.~(\ref{eq:bandwidth}). Note that, as $N\to\infty$, we recover the standard Kalman result from Eq.~(\ref{eq:Kalman}) as we should, since the kernel acts as a Dirac delta-function in Eq.~(\ref{eq:RPFMC}) as $\alpha_h\to0$. Additionnally, while we work in the limit of $N\to\infty$ and choose to conserve the $N$-dependence of the kernel bandwidth, we explicitly assume the Monte Carlo error $\propto N^{-1/2}$ is negligible when replacing the sum by the integral in Eq.~(\ref{eq:RPFMC}). We now proceed to illustrate the consequences of Eqs.~(\ref{eq:RPFMean},\ref{eq:RPFSigma}) by returning for simplicity to the case of stationary states, where exact expressions for the posterior moments can be obtained by solving Eqs.~(\ref{eq:RPFMean},\ref{eq:RPFSigma}) recursively. For reasons which will be clear later, we shall allow the possibility for $\alpha_h$ to be time-dependent. Note that our analysis should carry over to any dynamical system described by Eqs.~(\ref{eq:linHMM}) for which the Bayes estimator is asymptotically efficient: this includes systems with noiseless and compact supported state dynamics, and for which the observation model is regular in some sense \cite{CarmierKalman}.
\begin{lemma}
\label{lemma:general}
Consider the recursively defined set of equations for $n\geq1$ on the couple of variables $(\mu_n, \Sigma_n)$,
\be
\begin{array}{l}
\left(\Sigma_{n-1}^{-1}+R^{-1}\right){\bf \mu}_n = \Sigma_{n-1}^{-1}{\bf \mu}_{n-1} + R^{-1}{\bf y}_n
\vspace{0.1cm}
\\
\Sigma_n = (1+\alpha_n)R\left(R+\Sigma_{n-1}\right)^{-1}\Sigma_{n-1}
\end{array}
\ee
where $\{\alpha_n\}_{n\geq1}$ is a sequence of positive real numbers, and $R$ and $\{{\bf y}\}_{n\geq1}$ are defined as in Eq.~(\ref{eq:KalmanParam1}). Starting from the initial condition $(\mu_0, \Sigma_0)$, the solution to the above set of equations can be expressed as
\be
\label{eq:General}
\left|
\begin{array}{l}
\Sigma_n = \prod_{j=1}^n(1+\alpha_j)R\left( R + \sum_{j=1}^n\prod_{k=0}^{j-1}(1+\alpha_k)\Sigma_0 \right)^{-1}\Sigma_0
\vspace{0.1cm}
\\
\begin{split}
& {\bf \mu}_n = \frac{1}{\prod_{j=1}^n(1+\alpha_j)}\Sigma_n\Sigma_0^{-1}{\bf \mu}_0
\\
& + \left(\mathbb{I} - \frac{1}{\prod_{j=1}^n(1+\alpha_j)}\Sigma_n\Sigma_0^{-1}\right)\frac{\sum_{j=1}^n\prod_{k=0}^{j-1}(1+\alpha_k){\bf y}_j}{\sum_{j=1}^n\prod_{k=0}^{j-1}(1+\alpha_k)}
\end{split}
\end{array}
\right.
\ee
with $\alpha_0=0$.
\end{lemma}
\begin{proof}
This can be straightforwardly proven by recurrence.
\qed
\end{proof}
As a first application of lemma \ref{lemma:general}, consider the case where the sequence $\{\alpha_n\}_{n\geq1}$ is time-independent. This is our original motivation, corresponding to the choice of bandwidth indicated in Eq.~(\ref{eq:bandwidth}): $\alpha_n = \alpha_h \; \forall n\geq1$. Expressions for the posterior moments simplify as
\be
\label{eq:Kalpha}
\left|
\begin{array}{l}
\Sigma_n^N = (1+\alpha_h)^nR\left( R + \sum_{j=1}^n(1+\alpha_h)^{j-1}\Sigma_0 \right)^{-1}\Sigma_0
\vspace{0.1cm}
\\
\begin{split}
{\bf \mu}_n^N & = \frac{1}{(1+\alpha_h)^n}\Sigma_n^N\Sigma_0^{-1}{\bf \mu}_0
\\
& + \left(\mathbb{I} - \frac{1}{(1+\alpha_h)^n}\Sigma_n^N\Sigma_0^{-1}\right)\frac{\sum_{j=1}^n (1+\alpha_h)^{j-1}{\bf y}_j}{\sum_{j=1}^n (1+\alpha_h)^{j-1}}
\end{split}
\end{array}
\right.
\ee
We obtain as a direct consequence that, if $\alpha_h>0$ (for finite $N$), the covariance matrix does not converge to zero as the number of observations increases.
\begin{theorem}[RPF fixed point]
\label{th:bound}
The RPF estimator for the posterior covariance matrix converges exponentially fast towards a non-vanishing fixed point:
\be
\label{eq:bound}
\lim_{n\to\infty}\Sigma_n^N = \alpha_h R (1 + \mathcal{O}(e^{-n\alpha_h})) \; .
\ee
Likewise, the estimated mean value converges exponentially fast towards a random variable 
\be
\label{eq:const}
{\bf \mu}_n^N \underset{n\to\infty}{\approx} \frac{\sum_{j=1}^n (1+\alpha_h)^{j-1}{\bf y}_j}{\sum_{j=1}^n (1+\alpha_h)^{j-1}} \left( 1 + \mathcal{O}(e^{-n\alpha_h}) \right)
\ee
centered around the observation mean value ${\bf x}_n={\bf x}_0$, with residual covariance 
\be
\label{eq:xfluct}
\lim_{n\to\infty}\mathbb{E}[({\bf \mu}_n^N-{\bf x}_n)({\bf \mu}_n^N-{\bf x}_n)^T] = \frac{\alpha_h}{2+\alpha_h}R \; .
\ee
As a result, the asymptotic $\text{RMSE}=\mathcal{O}(\sqrt{\alpha_h})$ remains lower bounded by a positive constant as soon as the number of particles is finite, no matter how many observations.
\end{theorem}
As we have already mentioned, this situation can also occur in the original Kalman filter (retrieved in the limit of an infinite number of particles) as soon as process noise $Q\neq0$ is included. In the present case, the reason for this phenomenon is that the kernel bandwidth used for resampling acts as a perturbation which increases the variance and competes with the variance reduction due to each new observation. The problem is that the variance reduction rate with $\Sigma_n^N$ is faster (quadratic) than that of the bandwidth (linear): $\Sigma_{n+1}^N-\Sigma_n^N = (\alpha_h R - \Sigma_n^N)(R + \Sigma_n^N)^{-1}\Sigma_n^N$. Arrives therefore a moment when the ever smaller piece of information provided by the new observation is exactly compensated by the bandwidth perturbation. 

A complementary and enlightening interpretation of our results can be put forward by noting that each observation in Eq.~(\ref{eq:const}) contributes differently to the estimated posterior mean value, \textit{i.e.} each ${\bf y}_j$ carries a different weight. In particular, the weighting ratio between first (oldest) and last (most recent) observations roughly scales like $e^{-n\alpha_h}$, which signifies that the RPF estimator loses memory of previous observations exponentially fast. This loss of memory is directly responsible for the inability of the estimator to process the information provided by the full history of observations as exhaustively as the Kalman estimator, since it acts as an effective temporal resolution. The lower bound on the RMSE obtained in theorem \ref{th:bound} is thus simply the translation of this temporal resolution in state space.

Returning to Fig.~\ref{fig:1}, we check the validity of our results by comparing them with numerical simulations of the one-dimensional stationary HMM introduced in section \ref{sec:rr=1}. Interestingly, although conclusions of theorem \ref{th:bound} based on Eqs.~(\ref{eq:RPFMean},\ref{eq:RPFSigma}) were derived in the asymptotic limit $N\to\infty$, numerical simulations support their well-foundedness already for $N=1000$ particles. The late time fluctuations in the RMSE of the RPF correspond to those of mean value in Eq.~(\ref{eq:xfluct}). Next, because $\lim_{N\to\infty} \alpha_h = 0$, we expect according to theorem \ref{th:bound} the lower bound to decrease as $N$ increases. This is clearly illustrated by the blue curve in the left frame of Fig.~\ref{fig:1}. Increasing the observation noise variance obviously leads to the opposite effect (green curve). We plotted as a dashed black line the analytical prediction based on the Kalman filter from Eq.~(\ref{eq:Kalpha}), where the observation noise variance $R/\Sigma_0$ was quenched from the value 0.25 to 1 at $n=300$. As one can see, the computed posterior variance on the state reacts to the quench by increasing and converging towards the new asymptotic fixed point. Notice that increasing $d_x$ is also detrimental (not shown), since $\lim_{d_x\to\infty} \alpha_h = 1$ according to Eq.~(\ref{eq:bandwidth}). Next, we show in the right frame that other bandwidth selection criteria such as Silverman and Scott's rule of thumbs or the direct plug-in method by Sheather and Jones \cite{Sheather04} lead to the same problem when resampling is performed at each step. Another natural strategy is to perform resampling only occasionally. However, resampling periodically every $p$ step only allows reducing the location of the fixed point roughly a factor $1/p$, as is shown in the following corollary. 
\begin{corollary}
\label{cor:period}
Given a resampling period $p\in\mathbb{N}^*$ and an integer $0\leq q<p$, the RPF estimator for the posterior covariance matrix $\Sigma_{np+q}^N$ converges exponentially fast towards the non-vanishing q-dependent solution:
\be
\label{eq:bound_period}
\lim_{n\to\infty}\Sigma_{np+q}^N = \frac{\alpha_h R}{p+\alpha_hq} (1 + \mathcal{O}(e^{-n\alpha_h})) \; .
\ee
As a result, the asymptotic covariance matrix fluctuates between the upper bound $\alpha_hR/p$ and the lower bound $\alpha_hR/(p+(p-1)\alpha_h)$. Likewise, the estimated mean value converges exponentially fast towards a random variable centered around the observation mean value ${\bf x}_n$, with a q-dependent residual covariance
\be
\begin{split}
& \lim_{n\to\infty}\mathbb{E}[({\bf \mu}_{np+q}^N-{\bf x}_n)({\bf \mu}_{np+q}^N-{\bf x}_n)^T] 
\\
& = \frac{\alpha_h}{2+\alpha_h}\frac{p+\alpha_h(2+\alpha_h)q}{(p+\alpha_hq)^2}R \; .
\end{split}
\ee
This translates in an asymptotic $\text{RMSE}=\mathcal{O}(\sqrt{\alpha_h/p})$.
\end{corollary}
\begin{proof}
This can be proven by deriving from Eq.~(\ref{eq:General}) the new expression for the posterior covariance matrix when $\alpha_{np+q} = \alpha_h\delta_{q,0}$:
\be
\Sigma_{np+q}^N = (1+\alpha_h)^nR \left( R + \sum_{j=1}^{np+q} (1+\alpha_h)^{[j-1]_p}\Sigma_0 \right)^{-1}\Sigma_0 \; ,
\ee
where we defined $[j-1]_p$ as the integer part of $(j-1)/p$
. Similar calculations allow obtaining the corresponding expression for the mean value:
\be
\begin{split}
& {\bf \mu}_{np+q}^N = \frac{1}{(1+\alpha_h)^n}\Sigma_{np+q}^N\Sigma_0^{-1}{\bf \mu}_0
\\
& + \left(\mathbb{I} - \frac{1}{(1+\alpha_h)^n}\Sigma_{np+q}^N\Sigma_0^{-1}\right) \frac{\sum_{j=1}^{np+q} (1+\alpha_h)^{[j-1]_p}{\bf y}_j}{\sum_{j=1}^{np+q} (1+\alpha_h)^{[j-1]_p}} \; .
\end{split}
\ee
\qed
\end{proof}
That the RMSE remains lower bounded by a positive constant when resampling is performed periodically is illustrated numerically in the left frame of Fig.~\ref{fig:2}, where the analytical lower bound (\ref{eq:bound_period}) is indicated by the dotted line.
\begin{figure*}
\begin{center}
  \includegraphics[angle=0,width=0.45\linewidth]{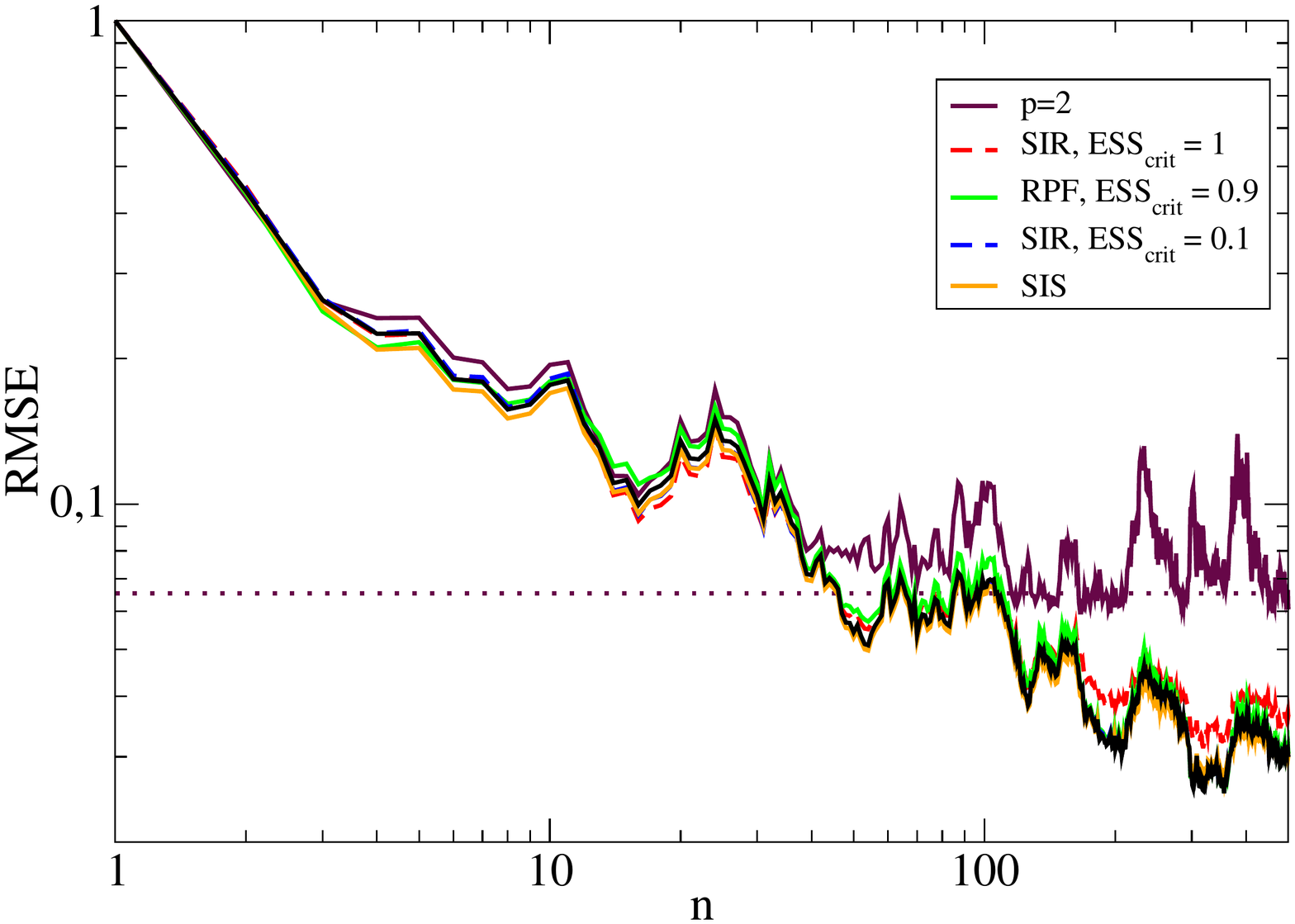}
  \includegraphics[angle=0,width=0.45\linewidth]{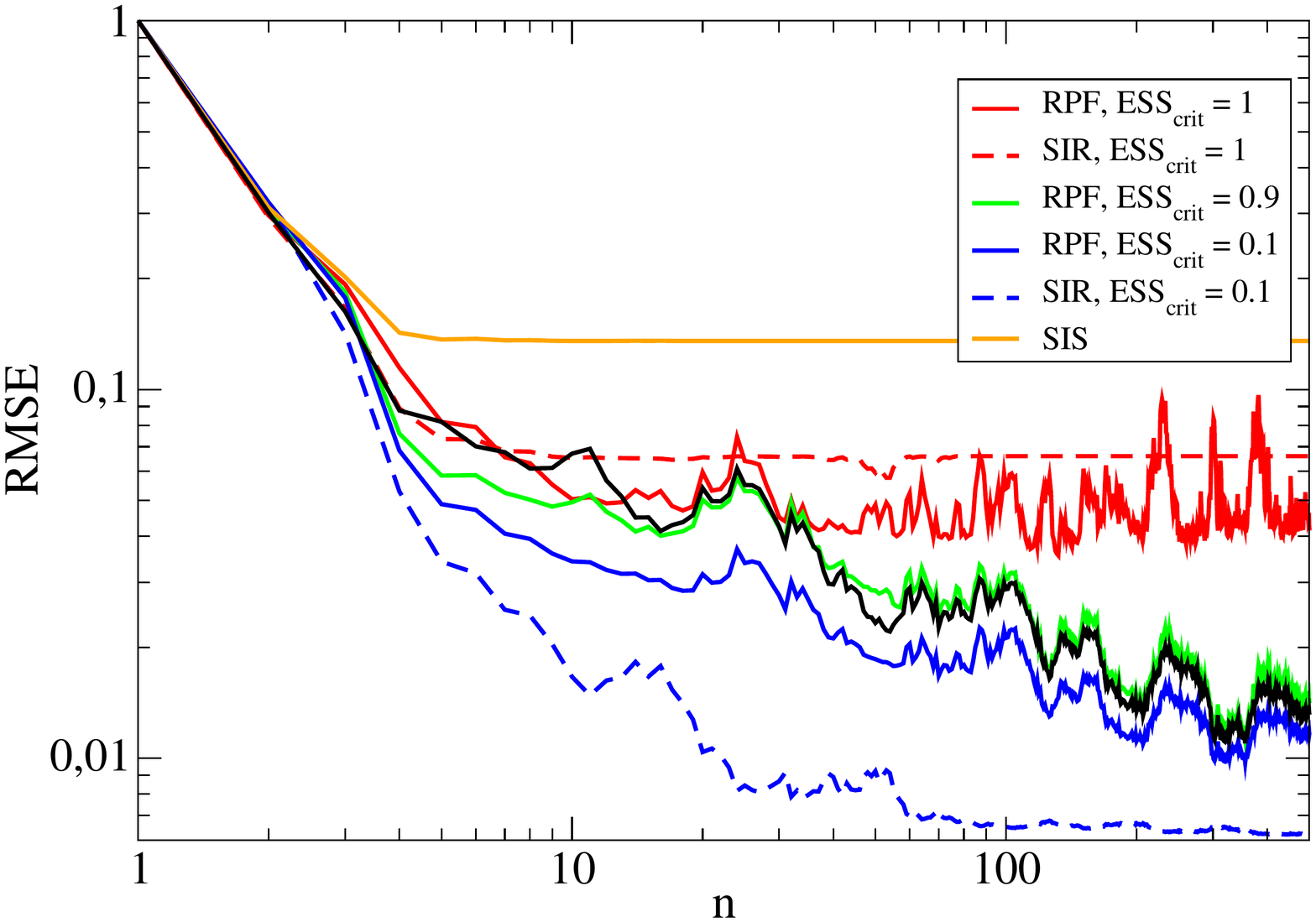}
\caption{Left frame: Comparison of RPF and SIR simulations for different values of the resampling ratio. Sequential importance sampling (SIS) is a special case of the SIR algorithm when no resampling is performed ($\text{ESS}_\text{crit}=0$). The Kalman filter expectation is provided for reference (solid black line). Estimation based on the RPF with $p=2$ periodic resampling is also displayed (solid maroon) and shown to saturate around the analytically predicted solution (dotted line). Right frame: Same but for a prior location $\mu_0$ three standard deviations away from the mean observation. This time, the SIR approximation breaks down, as opposed to the RPF which remains well-behaved.}
\label{fig:2}
\end{center}
\end{figure*} 

\subsection{Kernel bandwidth modulation}

One is thus left with two possibilities in order to allow the estimator to keep up with the Kalman filter prediction: either resample less and less often with time (which seems compatible with the idea of resampling only if the ESS falls lower than a critical threshold, see subsection \ref{sec:Stat}), or choose the bandwidth parameter differently so that consistency of the estimator is maintained despite performing resampling at each time. To begin with, it is clear from theorem \ref{th:bound} that choosing $\{\alpha_n\}_{n\geq1}$ such that $\lim_{n\to\infty}\alpha_n=0$ enforces that $\lim_{n\to\infty}\text{RMSE}(\hat{{\bf x}}_n, {\bf x}_n)=0$ (this can easily be proven by contradiction). This idea was actually hinted at in \cite{ChenPhD}. We proceed one step further by characterizing in theorem \ref{th:sequence} which sequences ensure that the rate of convergence of the RPF estimator is optimal in the sense defined below. 
\begin{definition}[Optimal convergence rate]
\label{def:Optimal}
The rate of convergence of an estimator $\hat{{\bf x}}_n$ towards the true state ${\bf x}_n$ as the number of observations $n\to\infty$ will be called optimal if it is comparable with that given by the inverse of the Fisher information matrix: $\text{RMSE}(\hat{{\bf x}}_n, {\bf x}_n)=\mathcal{O}(n^{-1/2})$.
\end{definition}
Since the Kalman filter estimator is asymptotically efficient, this implies in particular that the optimal rate of convergence is that given by the Kalman filter.
\begin{theorem}
\label{th:sequence}
Let $\{\alpha_n\}_{n\geq1}$ be a sequence of positive numbers as in Eq.~(\ref{eq:General}). If $\alpha_n=\mathcal{O}(n^{-1})$, the RPF estimator converges at an optimal rate. In contrast, given any $\epsilon\in]0, 1/2[$, if $\alpha_n\approx n^{\epsilon-1}$, then $\exists \; c \in \mathbb{R}^+ \; \forall n\geq1 \; \text{RMSE} \geq c n^{(\epsilon-1)/2}$.
\end{theorem}
\begin{proof}
Suppose that $\exists \; c \in \mathbb{R}^+ \; \forall n\geq1 \; n\alpha_n \leq c$. Then, using that $x-x^2/2 < \log{(1+x)} < x$ for $x>0$, one has that $\prod_{j=1}^n(1+\alpha_j)=\mathcal{O}(n^c)$, and thus also that $\sum_{j=1}^n \prod_{j=0}^{k-1}(1+\alpha_j)=\mathcal{O}(n^{c+1})$. Plugging this in Eq.~(\ref{eq:General}) implies that $\Sigma_n^N=\mathcal{O}(R/n)$ and therefore that the RPF estimator converges at an optimal rate. In contrast, if $\alpha_n\approx n^{\epsilon-1}$, then the same reasoning implies that $\prod_{j=1}^n(1+\alpha_j)=\mathcal{O}(\exp{(n^\epsilon/\epsilon)})$. Making use of the following integral inequality
\be
\begin{split}
n^{\epsilon-1}\sum_{j=0}^{n-1}\exp{(\frac{j^\epsilon}{\epsilon})} & \leq \sum_{j=0}^{n-1}(j+1)^{\epsilon-1}\exp{(\frac{j^\epsilon}{\epsilon})}
\\
& \leq \int_0^n dx \; x^{\epsilon-1} \exp{(\frac{x^\epsilon}{\epsilon})} \leq \exp{(\frac{n^\epsilon}{\epsilon})} \; ,
\end{split}
\ee
we deduce that $\exists \; c \in \mathbb{R}^+ \; \forall n\geq1 \; \Sigma_n^N\geq c n^{\epsilon-1}R$. Finally, noting that $\sum_{j=0}^{n-1}\exp{(2j^\epsilon/\epsilon)}\geq \exp{(2(n^\epsilon-1)/\epsilon)}$, we conclude that $\exists \; c \in \mathbb{R}^+ \; \forall n\geq1 \; \text{RMSE}\geq c n^{(\epsilon-1)/2}$.
\qed
\end{proof}
This result still leaves us with a wide range of sequences at our disposal. Finding a suitable statistical criterion both justifying this type of kernel bandwidth modulation (and therefore different from the mean integrated squared error used to derive Eq.~(\ref{eq:bandwidth})) and allowing to identify a sequence which is optimal in some sense remains an open question. Let us look at two examples for the sequence $\{\alpha_n\}_{n\geq 1}$.
\begin{itemize}
\item[$\bullet$] Example 1: $\alpha_n = (n+\alpha_h^{-1})^{-1}$.
\end{itemize}
For this choice, products and sums in Eq.~(\ref{eq:General}) can be computed exactly and we get for the covariance matrix $\Sigma_n^N \underset{n\to\infty}{\approx} 2R/n$, and for the mean value
\be
\label{eq:power}
{\bf \mu}_n^N \underset{n\to\infty}{\approx} \frac{2}{n^2}\sum_{j=1}^n j{\bf y}_j ( 1 + \mathcal{O}(n^{-2}) )
\ee
centered around the observation mean value ${\bf x}_n$ and with residual covariance $\mathbb{E}[({\bf \mu}_n^N-{\bf x}_n)({\bf \mu}_n^N-{\bf x}_n)^T] \underset{n\to\infty}{\approx} 4R/(3n)$. Thus, optimal convergence rate is recovered, as expected from theorem \ref{th:sequence}. Interestingly, as compared to Eq.~(\ref{eq:const}), the dependence of Eq.~(\ref{eq:power}) on the observations is long-ranged: the weighting ratio between the last and the first observations scales linearly with $n$ instead of exponentially as in Eq.~(\ref{eq:const}). This illustrates that consistency of the regularized estimator requires that memory of past observations should be kept long enough. Note that such a requirement, while understandable for parameter estimation, may be irrelevant for tracking applications where memory of the past might be considered relatively useless.
\begin{itemize}
\item[$\bullet$] Example 2: $\alpha_n = \alpha_he^{-n\alpha_h}$.
\end{itemize}
In this case, one cannot compute sums and products in Eq.~(\ref{eq:General}) exactly but one can easily show that $\prod_{j=1}^n(1+\alpha_he^{-j\alpha_h})$ converges to a non-zero limit using that it is a strictly increasing bounded sequence:
\be
1 \leq \prod_{j=1}^n(1+\alpha_he^{-j\alpha_h}) \leq \exp{(\frac{\alpha_h}{e^{\alpha_h}-1})} \; .
\ee
This result allows immediately deducing that $\sum_{j=1}^n\prod_{k=0}^{j-1}(1+\alpha_he^{-k\alpha_h}) = \mathcal{O}(n)$, and thus that $\Sigma_n^N = \mathcal{O}(R/n)$ and that the mean value
\be
{\bf \mu}_n^N \underset{n\to\infty}{\approx} \frac{\sum_{j=1}^n \prod_{k=0}^{j-1}(1+\alpha_he^{-k\alpha_h}){\bf y}_j}{\sum_{j=1}^n \prod_{k=0}^{j-1}(1+\alpha_he^{-k\alpha_h})} \left( 1 + \mathcal{O}(n^{-1}) \right)
\ee
is centered around the observation mean value ${\bf x}_n$ with residual covariance $\mathbb{E}[({\bf \mu}_n^N-{\bf x}_n)({\bf \mu}_n^N-{\bf x}_n)^T] = \mathcal{O}(R/n)$, using that $\prod_{k=0}^{j-1}(1+\alpha_he^{-k\alpha_h})^2$ is also bounded. Optimal convergence rate is thus once more recovered, this time with a weighting ratio between different observations bounded by $\prod_{j=1}^\infty(1+\alpha_he^{-j\alpha_h})$. 

This strategy of modulating the bandwidth by enforcing it to decrease fast enough is positively supported by numerical simulations in the right frame of Fig.~\ref{fig:1}. It thus constitutes a functional method to overcome the limitations in the RPF we discussed earlier. Note that a different strategy discussed in Ref.~\cite{West93} consists in counterbalancing the variance originating from the bandwidth perturbation by rescaling the particle values towards their center of mass. Centering the kernel densities around the shrunken values $a{\bf x}_{n|n-1}^{(i)} + (1-a){\bf \mu}_n^N$, and taking $a=\sqrt{1-\alpha_h}$, one can indeed verify that the posterior density in the RPF asymptotically coincides with that expected from the Kalman filter approach.

\subsection{Strategies based on the resampling ratio}
\label{sec:Stat}

We now examine another possibility, which 
is to perform resampling only if the ESS introduced in Eq.~(\ref{eq:ESS}) becomes lower than a determined threshold which we shall refer to as the resampling ratio $\text{ESS}_\text{crit}$ \cite{DelMoral12}. In the left frame of Fig.~\ref{fig:2}, we compare the posterior RMSE for different values of this threshold, both for the RPF and the SIR algorithm. As one can clearly see, as soon as the resampling ratio is smaller than 1, the RPF yields results in agreement with the exact solution. The SIR algorithm provides comparable results, whatever the value of the resampling ratio. However, if the location $\mu_0$ of the prior density is taken farther away from the correct value, the situation changes since estimating the posterior mean value $\mu_n$ accurately will increasingly depend on the ability of particles to explore phase space well enough to diffuse towards the correct value. This is illustrated in the right frame of Fig.~\ref{fig:2}, where $\mu_0$ was taken to be $3\sqrt{\Sigma_0}$ away from the observed mean value, instead of $\sqrt{\Sigma_0}$ until now. In this case, the SIR algorithm almost always fails to correctly estimate the posterior variance due to sample impoverishment arising as a direct consequence of severe weight degeneracy. In fact, the estimated variance often goes to zero, reflecting the fact that only a single particle eventually survives the resampling process. On the other hand, the RPF proves very robust in this situation and provides correct results as before, despite the fact that its RMSE has a tendency to underestimate the optimal Kalman expectation. Note that the curves displayed on Fig.~\ref{fig:2} are outputs of single particle filter simulations (no averaging is performed).

Resampling conditional on a sub-unit value of $\text{ESS}_\text{crit}$ thus also appears to resolve the RPF shortcomings discussed earlier. One can actually show that the resampling frequency for a sub-unit resampling ratio decreases exponentially fast with the number of observations, thereby justifying that the posterior RMSE should no longer be lower bounded by a positive constant. We provide below a formal proof of this statement in the case of the stationary HMM.
\begin{lemma}
\label{th:ESS}
Assume that $\text{ESS}_{\text{crit}}<1$. Then the spacing between consecutive resampling times increases exponentially fast with the number of observations for a stationary HMM.
\end{lemma}
\begin{proof}
Our starting point is the observation that
\be
\begin{split}
& \lim_{N\to\infty}\frac{\text{ESS}_n}{N} = \lim_{N\to\infty}\frac{\left(\frac{1}{N}\sum_{i=1}^N \prod_{j=1}^n g_j({\bf y}_j|{\bf x}_{j|j-1}^{(i)}) \right)^2}{\frac{1}{N}\sum_{i=1}^N \prod_{j=1}^n g_j({\bf y}_j|{\bf x}_{j|j-1}^{(i)})^2}
\\
& = \frac{\left(\int_{{\cal X}^{n+1}} \pi_0(d{\bf x}_0) \prod_{j=1}^n g_j({\bf y}_j|{\bf x}_j)f_j(d{\bf x}_j|{\bf x}_{j-1}) \right)^2}{\int_{{\cal X}^{n+1}} \pi_0(d{\bf x}_0) \prod_{j=1}^n g_j({\bf y}_j|{\bf x}_j)^2f_j(d{\bf x}_j|{\bf x}_{j-1})}  \; .
\end{split}
\ee
For simplicity, we consider the case of the one-dimensional stationary state, where integration with respect to the Markov kernels reduces to $\prod_{j=1}^n g_j(y_j|x_0)$. It is then a simple task to compute the effective sample size:
\be
\label{eq:ESSgaussian}
\begin{split}
\lim_{N\to\infty}\frac{\text{ESS}_n}{N} & = \frac{\sqrt{\gamma_n(2 + \gamma_n)}}{1+\gamma_n}\exp{\left(-\frac{\langle \mu_0 - y\rangle_n^2}{\Sigma_0(1+\gamma_n)(2+\gamma_n)}\right)} \\
& \underset{n\to\infty}{\approx} \sqrt{2\gamma_n} \exp{\left(-\frac{\langle \mu_0 - y\rangle_n^2}{2\Sigma_0}\right)} \; .
\end{split}
\ee
We see that $\text{ESS}_n$ depends sensitively on how close the ML $\langle y\rangle_n$ is from the typical particle value $\mu_0$ and on the quantity $\gamma_n := R/(n\Sigma_0)$. The number of observations essentially reduces the effective observation noise variance. If resampling is performed at step $n$, then one can express the ESS at step $n+m$ asymptotically as 
\be
\frac{\text{ESS}_{n+m}}{N} \underset{n, m\to\infty}{\approx} \frac{\sqrt{n(2m+n)}}{m+n} \; . 
\ee
The value of $m$ required for the ESS to fall below a critical threshold $\text{ESS}_\text{crit}N$, with $\text{ESS}_\text{crit} \in [0, 1]$, can then easily be calculated and we get
\be
\frac{\text{ESS}_{n+m}}{N} <  \text{ESS}_\text{crit} \Leftrightarrow \frac{m}{n} > \beta_{\text{crit}}
\ee
with
\be
\beta_{\text{crit}} = \frac{\sqrt{1-\text{ESS}_\text{crit}^2}}{1-\sqrt{1-\text{ESS}_\text{crit}^2}} \; .
\ee
Iterating, we thus see that the number of steps $m_k$ asymptotically required to perform resampling for the $k^{\text{th}}$ time grows exponentially fast: 
\be
\begin{split}
m_k & = \beta_{\text{crit}}(n+\sum_{j=1}^{k-1}m_j) = (1+\beta_{\text{crit}})m_{k-1} 
\\
& = \beta_{\text{crit}}(1+\beta_{\text{crit}})^{k-1}n \; .
\end{split}
\ee
\qed
\end{proof}
Note that similar conclusions were reached in \cite{Chopin02}. We provide numerical simulations supporting this claim in Fig.~\ref{fig:ESS}.
\begin{figure}
\begin{center}
  \includegraphics[angle=0,width=1.1\linewidth]{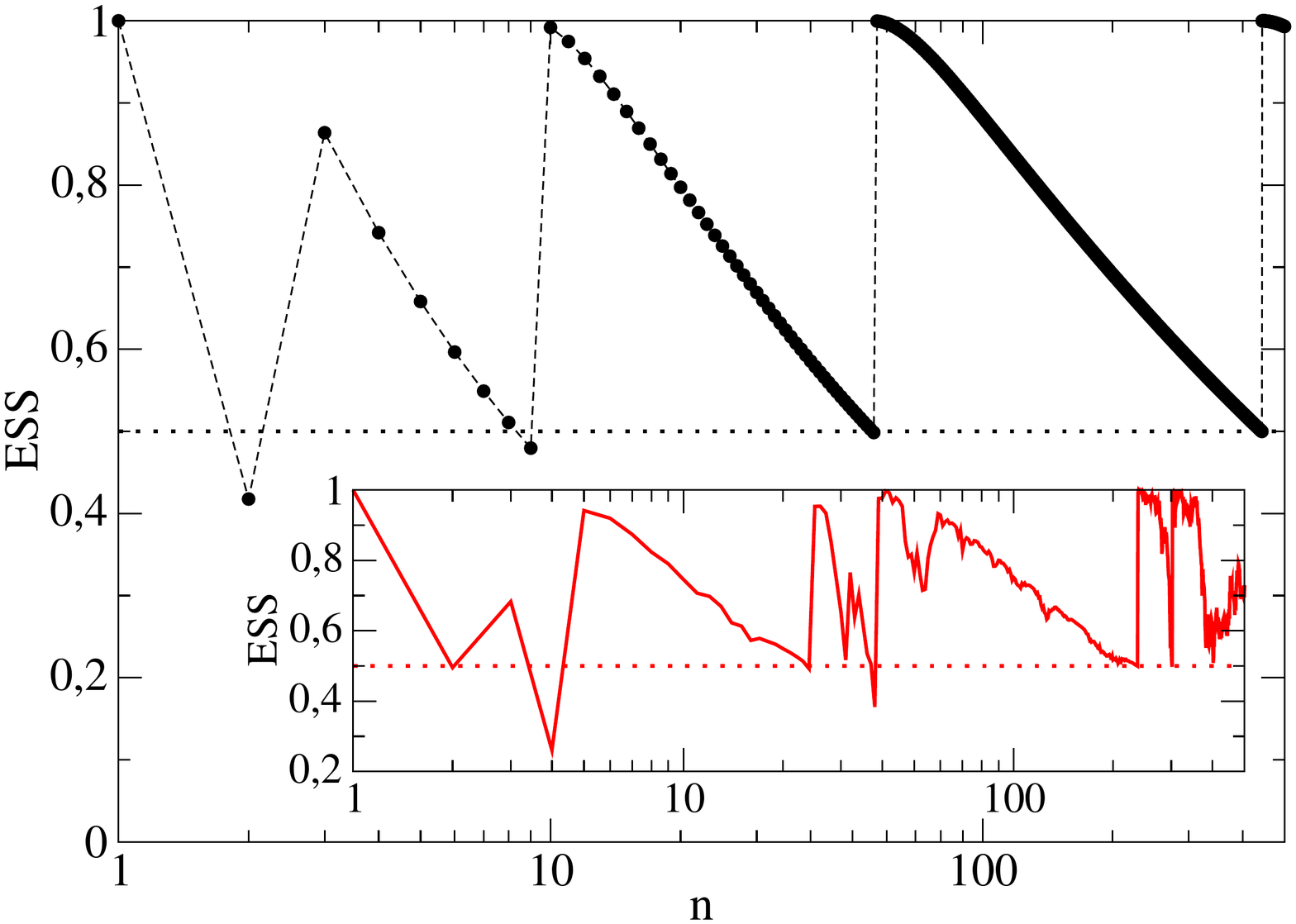}
\caption{Main: Normalized effective sample size (ESS) for $N=1000$ particles as a function of the number of noiseless observations $n$ in logarithmic scale for the stationary HMM. The exponential growth of the characteristic time required for the ESS to fall below a given threshold (here $\text{ESS}_\text{crit}=0.5$ indicated by the dashed line) is manifest. Inset: Same but for noisy observations. This time the exponential growth is somewhat less apparent due to the random fluctuations.}
\label{fig:ESS}
\end{center}
\end{figure} 
Interestingly, we compared the results we obtain for noisy observations (inset) with that for noiseless observations (symbols in main frame). The latter can be interpreted as representing the most favourable scenario, in which observations are always equal to the noiseless signal, and shall therefore be referred to from now on as the oracle limit, by analogy with the influential decision-theoretic concept \cite{Candes06}. As one can notice in the inset, the exponential growth of the resampling step sizes can be occasionally violated if observations cause large fluctuations in the exponential term of Eq.~(\ref{eq:ESSgaussian}). These fluctuations also explain why the $\text{ESS}$ is not necessarily monotonous between consecutive resampling times for noisy observations.

Equipped with lemma \ref{th:ESS}, we can now conclude our demonstration that the RPF estimator is well-behaved as soon as $\text{ESS}_\text{crit}<1$.
\begin{theorem}
Assume that $\text{ESS}_{\text{crit}}<1$. Then the rate of convergence of the RPF estimator is optimal in the sense of definition \ref{def:Optimal}.
\end{theorem}
\begin{proof}
According to the previous lemma, we know that the spacing between consecutive resampling times increases exponentially fast, with rate $1+\beta_\text{crit}$. As a consequence,
\be
\prod_{j=1}^n(1+\alpha_j) = \mathcal{O}\left((1+\alpha_h)^{\frac{\log{n}}{\log{(1+\beta_\text{crit})}}}\right) = \mathcal{O}\left(n^{\frac{\log{(1+\alpha_h)}}{\log{(1+\beta_\text{crit})}}}\right) \; ,
\ee
and thus $\exists \; c \in \mathbb{R}^+ \; \forall n\geq1\; \prod_{j=1}^n(1+\alpha_j) = \mathcal{O}(n^c)$. The same reasoning as that followed in the proof of theorem \ref{th:sequence} allows us to conclude.
\qed
\end{proof}

\section{Extension to non-linear models}
\label{sec:NL}

We now provide numerical evidence that the above results qualitatively hold in the case of non-linear models. We shall consider two examples: the logistic map, which is notably famous for possessing regions of parameter space displaying chaos, and a multi-dimensional plant growth model.

\begin{figure}
\begin{center}
  \includegraphics[angle=0,width=1.1\linewidth]{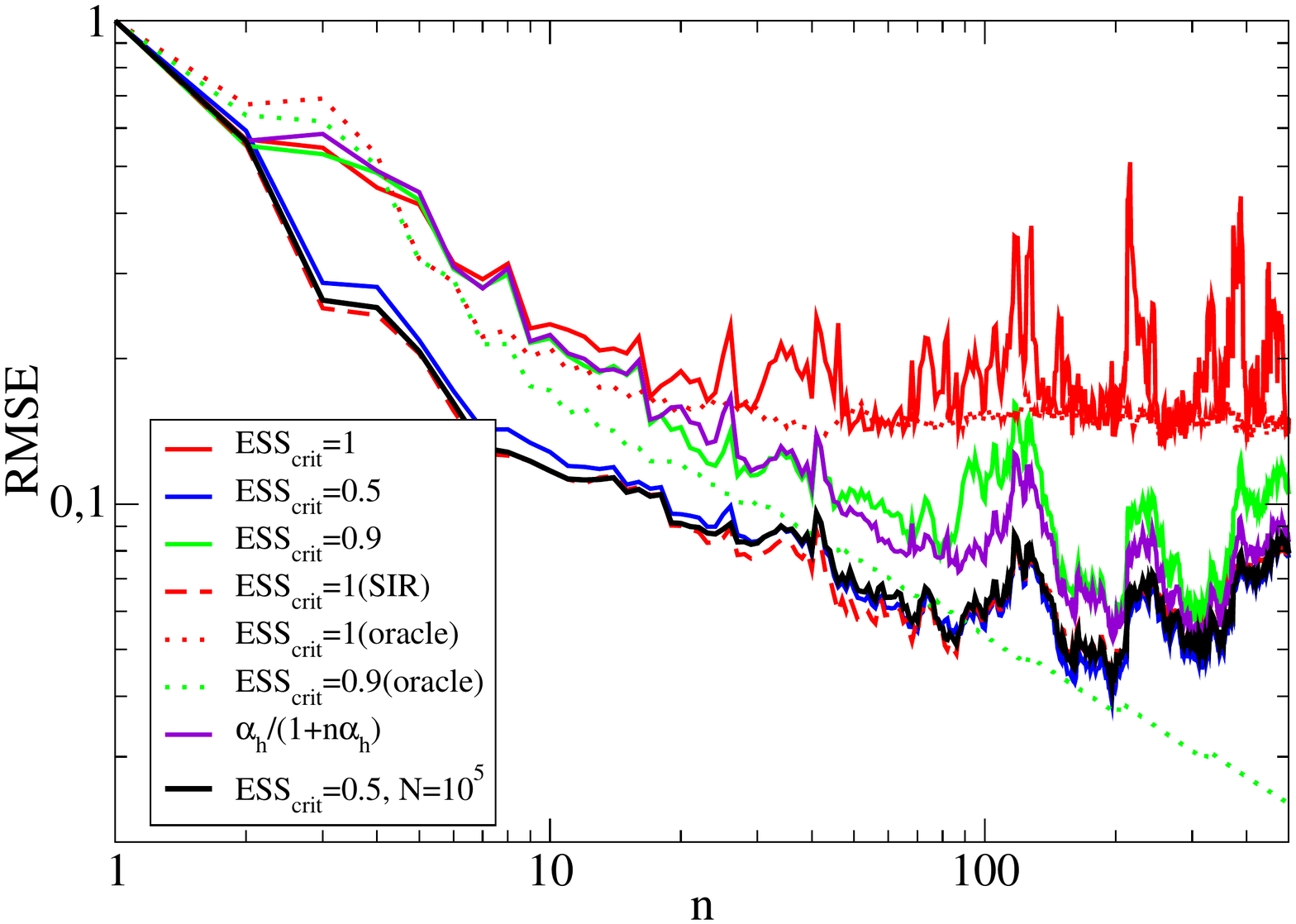}
\caption{Normalized root mean squared error of the estimated parameter $a$ in the logistic map as a function of the number of observations in logarithmic scale for $N=1000$ particles and normalized observation noise variance $R/\Sigma_0=0.33$. The saturation in the RMSE (solid red) testifies that the effect predicted for linear models remains valid in non-linear cases. Other strategies comparably improve the estimation.}
\label{fig:logistic}
\end{center}
\end{figure}

\subsection{Logistic map} 

\begin{figure*}
\begin{center}
  \includegraphics[angle=0,width=0.45\linewidth]{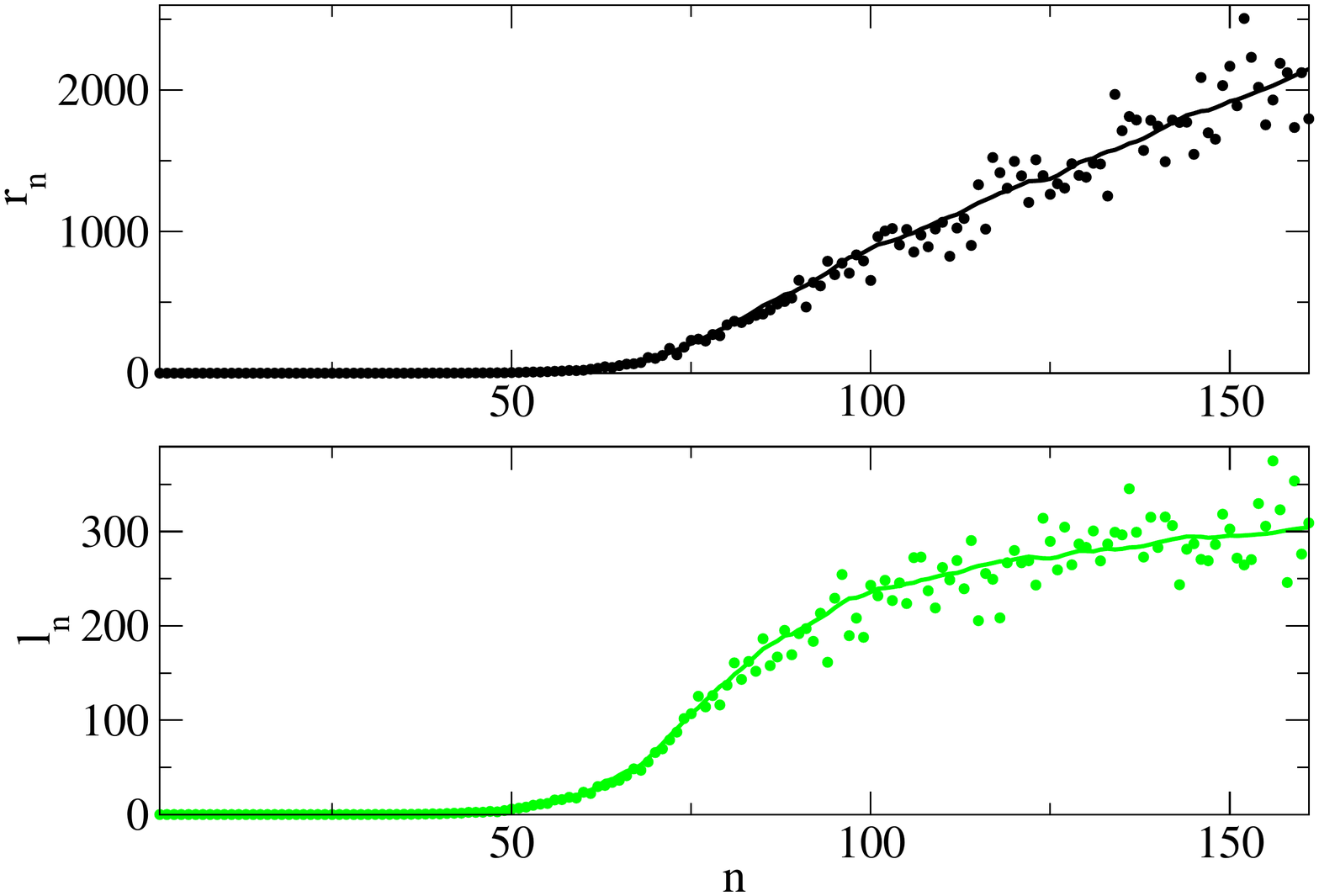}
  \includegraphics[angle=0,width=0.45\linewidth]{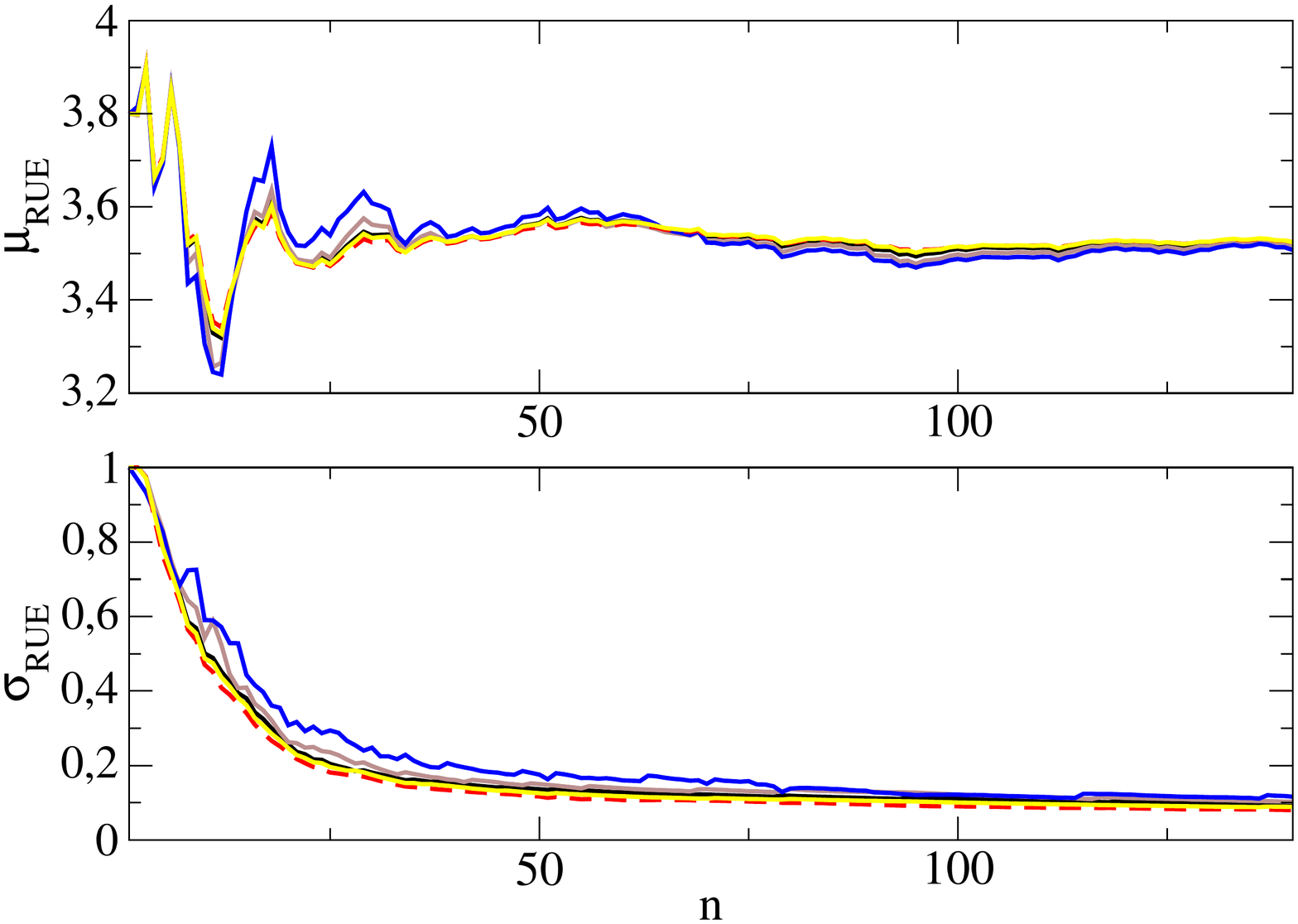}
\caption{Upper left: Observations of sugarbeet root mass $r_n$ (noisy = symbols, noiseless = thick line). Lower left: Likewise for leaf mass $l_n$. Upper right: Evolution of the mean value of the estimated parameter $\text{RUE}$ in the LNAS model after each new observation at time $n$, for $\text{ESS}_\text{crit}=0.5$ and various particle numbers: $N=1000$ (blue), $N=10^4$ (brown), $N=10^5$ (black) and $N=10^6$ (yellow). Convergence of the particle filter estimation towards the true posterior mean value is apparent, as a comparison with a SIR simulation testifies ($N=10^6$, dashed red). Lower right: Likewise for the normalized standard deviation.}
\label{fig:lnas}
\end{center}
\end{figure*}

The logistic map, popularized by Robert May \cite{May76}, is a discretized version of the logistic equation first introduced by Verhulst in the 19$^{\text{th}}$ century to describe population dynamics with a finite carrying capacity. It is described by the deceptively simple equation $x_{n+1} = ax_n(1-x_n)$, with ${\cal X}=[0,1]$ for $a\in[0,4]$. The nature of the state dynamics can change drastically depending on the value of $a$: it converges to a stationary solution for $0\leq a\leq 3$, oscillates between multiple values for $3<a\leq a_c$, and displays chaotic behavior for $a>a_c$, where $a_c\simeq 3.57$, with islands of stability found amid. The effect of the initial condition, on the other hand, is often of litte relevance and shall be assumed known to us: $x_0=0.5$. In the following, we seek to estimate $a$ given multiplicative noisy observations $y_n = x_n\eta$, with $\eta \sim \log{\cal N}(1, R)$. Starting from a prior distribution $a^{(i)} \sim {\cal N}(3.0, 0.3)$, the resulting inference of parameter $a$ given observations generated with the value $a^*=3.33$ and $R=0.1$ is displayed in Fig.~\ref{fig:logistic}.
Saturation of $\text{RMSE}(\hat{a}_n, a^*)$ when resampling is performed at each step is clearly apparent (solid red). This effect disappears using either of the resampling strategies we discussed earlier, however, as the comparison with a benchmark simulation for $N=10^5$ particles (solid black) testifies. Interestingly, we compared the results we obtain for noisy observations (thick lines) with that in the oracle limit (dotted lines). The RMSE obtained in the latter case acts as a lower bound with respect to that obtained for noisy observations, with numerical simulations even indicating that $\lim_{n\to\infty}\text{RMSE}(\hat{a}_n,a^*)=0$ in this case for $\text{ESS}_\text{crit}<1$.

If instead we choose to work in the chaotic parameter region and take $a^*>a_c$, our numerical results (not shown) seemingly point to the fact that correctly inferring the value of $a^*$ is always challenging, even in the oracle limit. The same can be told for the estimation of the initial condition $x_0$, for a fixed known value of $a$. The reason for this behavior seems to stem from the fact that, in such situations, the information provided by observations on the quantity of interest is a (rapidly) decreasing function of time. This loss of memory on the initial condition ultimately limits the resolution with which inference can be carried, as concentration results recently obtained in a similar context may testify \cite{Paulin16}.

\subsection{LNAS model}

Next, we look at a more elaborate model, which describes the growth of leaf mass $l_n$ and root mass $r_n$ of a sugarbeet plant as a function of daily temperature $T_n$ and radiation $\phi_n$. The latter constitute environmental variables, while the former are the hidden states. This model is a simplified version of the LNAS model \cite{Cournede13}. In particular, we consider deterministic transition functions. Leaf and root masses are updated as
\be
\left( \begin{array}{l} l_n \\ r_n \end{array} \right) = \left( \begin{array}{l} l_{n-1} \\ r_{n-1} \end{array} \right) + \left( \begin{array}{l} a_n \\ 1-a_n \end{array} \right)q_n \; ,
\ee
where $q_n$ is the daily produced biomass and $a_n$ is an allocation function. Biomass production is governed by the Beer-Lambert law
\be
\label{eq:BL}
q_n = \text{RUE} \; \phi_n \left( 1 - e^{-\rho^{-1} l_n} \right) \; ,
\ee
where $\text{RUE}$ is the radiation use efficiency and $\rho^{-1}$ the leaf mass density. Allocation depends on the thermal time $\tau_n$ as
\be
a_n = \frac{\gamma}{2}\left(1 - \text{erf}[\frac{1}{\sqrt{2}\sigma_a}\log{(\frac{\tau_n}{\mu_a})}]\right) \; ,
\ee
where $\gamma$ controls early biomass partition, $\text{erf}$ is the error function and $(\mu_a, \sigma_a)$ are parameters governing the shape of the allocation function. Thermal time increases monotonically as a function of temperature as $\tau_n = \tau_{n-1} + \text{max}[0, T_n]$. The system is initialized as $\tau_0=r_0=0$, and $l_0\neq 0$ such that $\rho^{-1} l_0 \ll 1$. For environmental variables, we use real data obtained from a weather station in France. The filtering problem in the present case then consists in estimating the parameters ($\text{RUE}$, $\gamma$, $\mu_a$), given independent prior densities $\text{RUE}^{(i)} \sim {\cal N}(3.8, 0.3)$, $\gamma^{(i)} \sim {\cal N}(0.7, 0.05)$ and $\mu_a^{(i)} \sim {\cal N}(500, 50)$, and given multiplicative noisy observations of the hidden states $l_n\eta$ and $r_n\eta$ with $\eta \sim \log{\cal N}(1,R)$. We shall assume $\rho$ and $\sigma_a$ known beforehand. 
\begin{figure*}
\begin{center}
  \includegraphics[angle=0,width=0.33\linewidth]{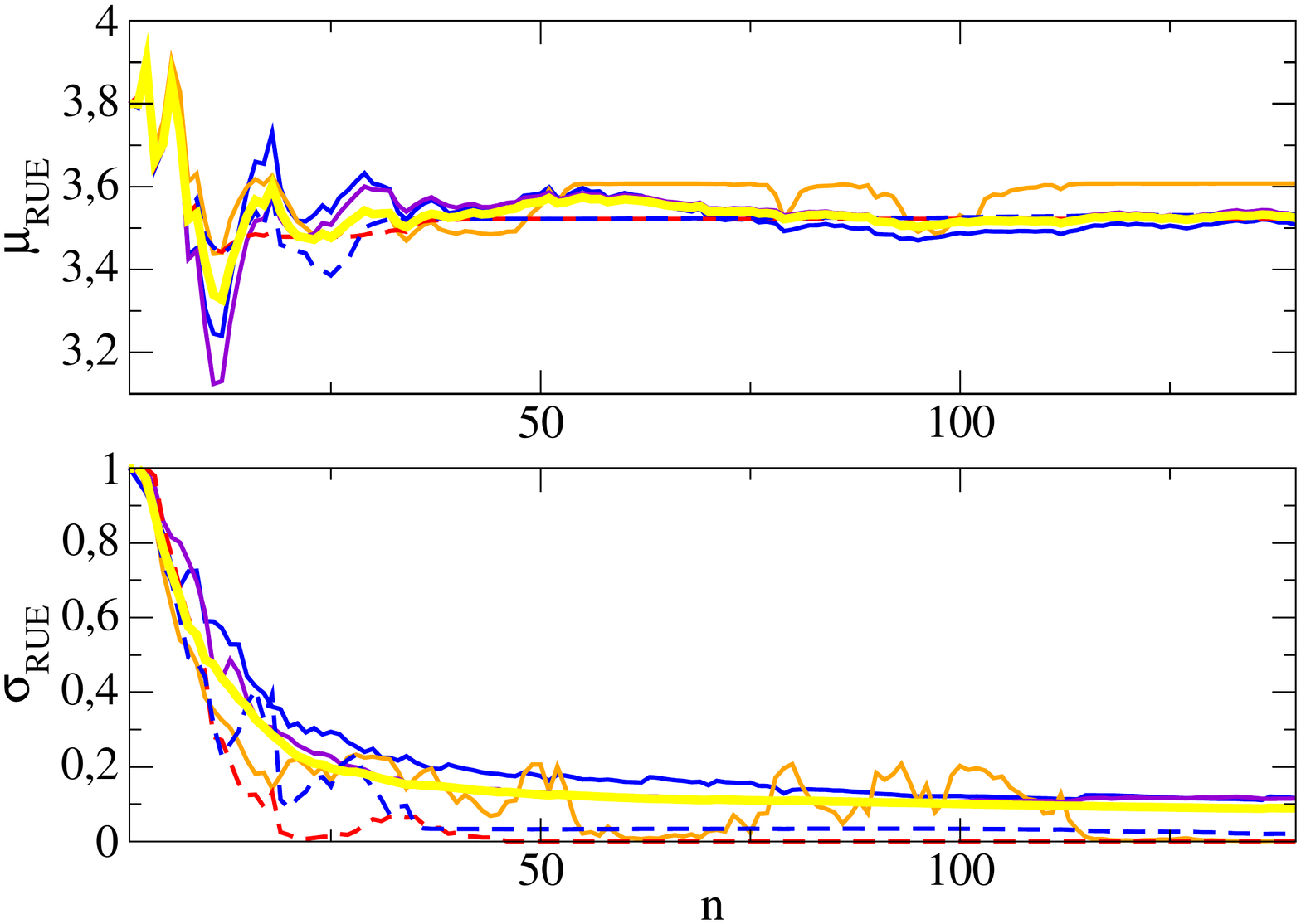}
  \includegraphics[angle=0,width=0.33\linewidth]{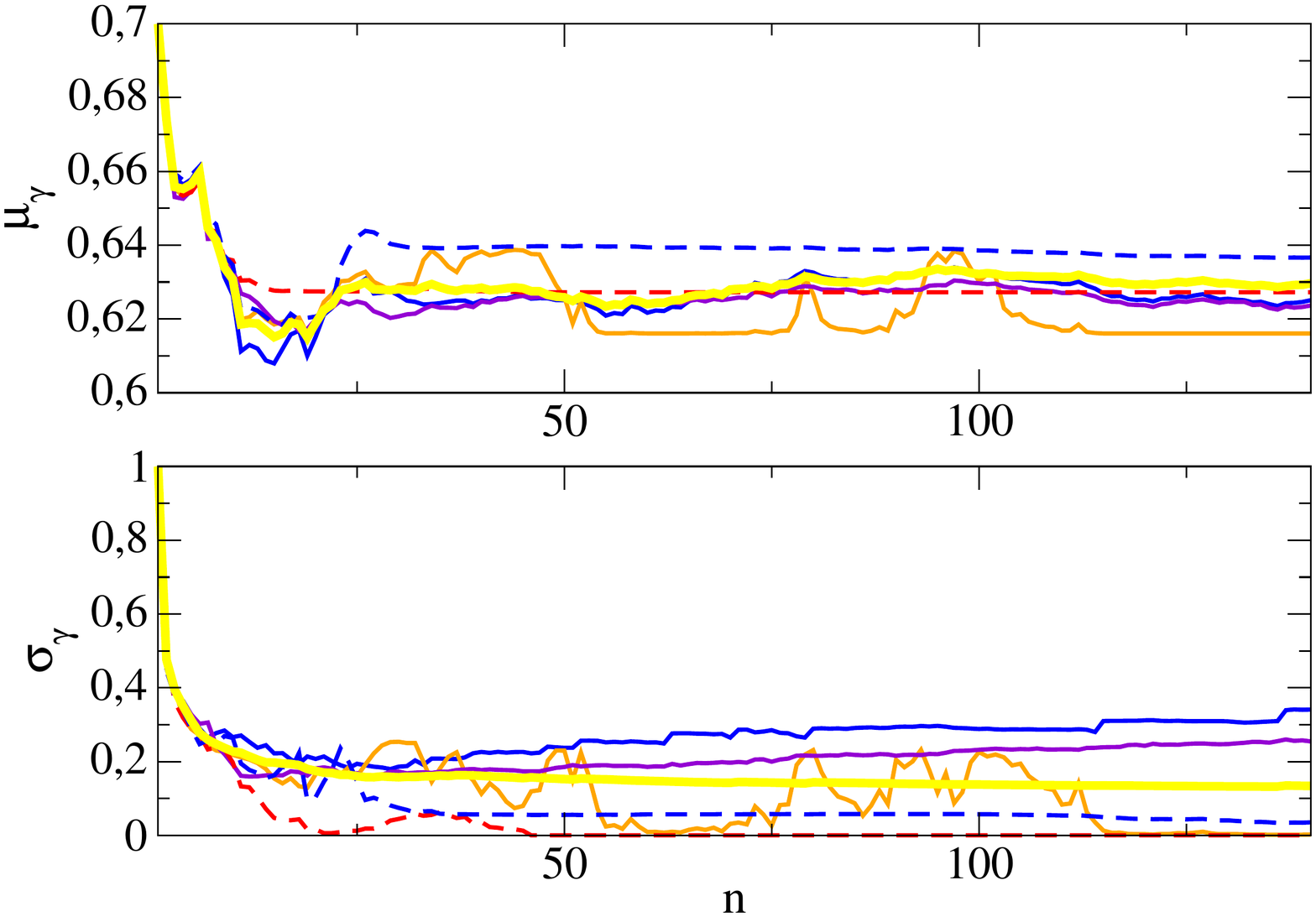}
  \includegraphics[angle=0,width=0.33\linewidth]{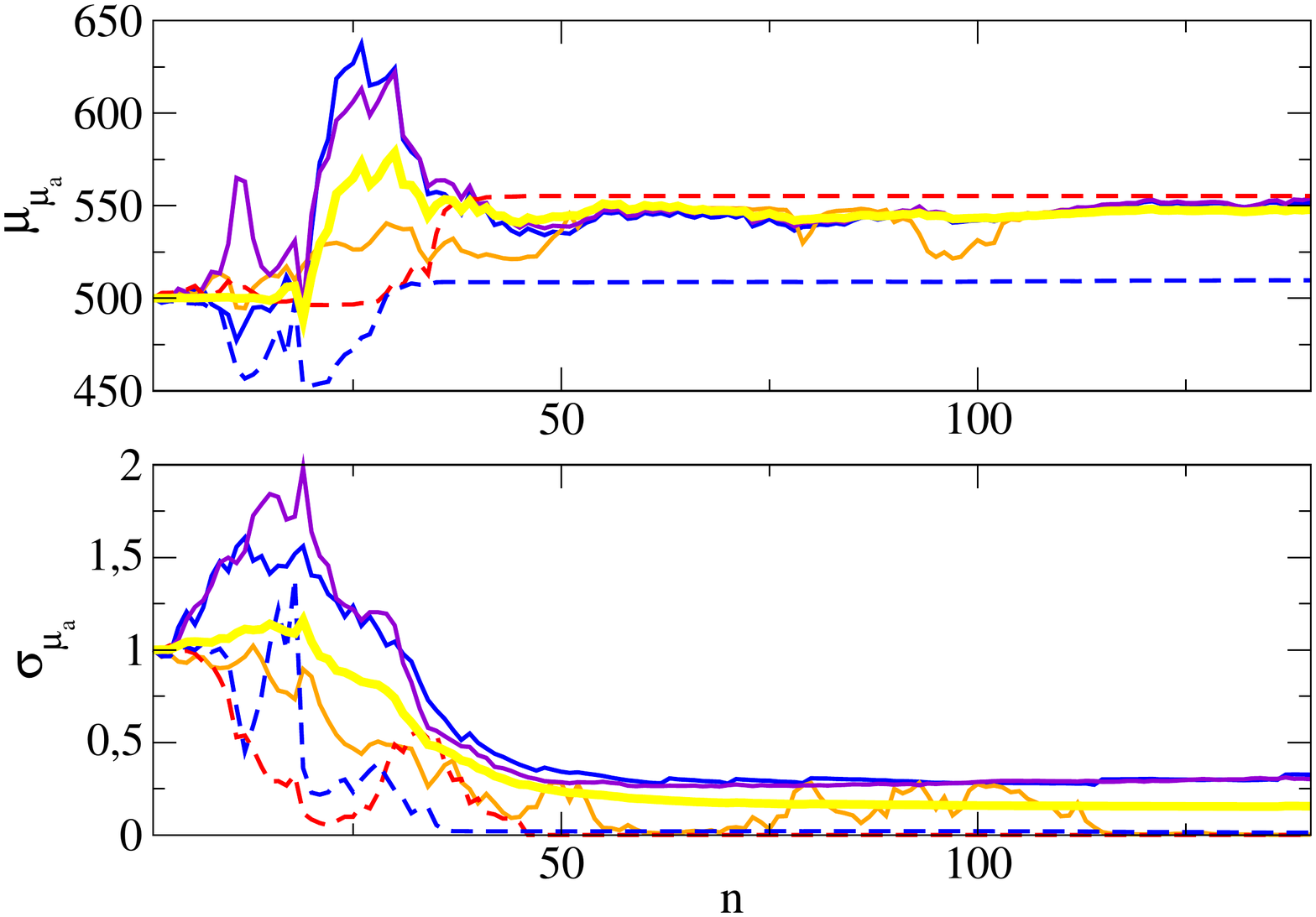}
\caption{Upper left: Evolution of the mean value of the estimated parameter $\text{RUE}$ in the LNAS model after each new observation at time $n$, for $N=1000$ particles and different resampling strategies: ESS$_\text{crit}=0.5$ (blue), $\alpha_n=\alpha_h/(1+n\alpha_h)$ (violet), ESS$_\text{crit}=0$ (orange). Lower left: Likewise for the normalized standard deviation. The bootstrap particle filter expectations (ESS$_\text{crit}=1.0$ (dashed red), ESS$_\text{crit}=0.5$ (dashed blue)) are clearly flawed with respect to the benchmark limit ($N=10^6$ particles, solid yellow), while the RPF provides excellent results already for $N=1000$ particles. Middle: Likewise for $\gamma$. Right: Likewise for $\mu_a$.}
\label{fig:lnas2}
\end{center}
\end{figure*}

Because we no longer have access to the optimal posterior density, as opposed to the case of the stationary model, we plot the evolution of posterior standard deviations and mean values of each parameter (instead of the RMSE) and use results obtained by the RPF for $N=10^6$ particles and ESS$_\text{crit}=0.5$ as a benchmark. We checked that this corresponds to the optimal posterior density by plotting the results obtained for $N=10^6$ particles using the SIR algorithm by comparison, and indeed observed that convergence as a function of $N$ is achieved (see Fig.~\ref{fig:lnas}). These simulations were carried by running our algorithms on the computing mesocenter Fusion. Given observations generated with values $(\text{RUE}^*,\gamma^*,\mu_a^*)=(3.56,0.625,550)$ and $R=0.1$, the resulting inference for various numbers of particles and resampling strategies is displayed in Fig.~\ref{fig:lnas2}. We see that the performance of the RPF for $N=1000$ particles and ESS$_\text{crit}=0.5$ is already close to that obtained in the benchmark scenario. Using our kernel bandwidth modulation strategy, we additionally demonstrate that equivalent results can be obtained even when ESS$_\text{crit}=1$, whereas the traditional choice of Eq.~(\ref{eq:bandwidth}) leads to an unstable exploration of parameter space (not shown). 
This holds independently of the chosen parameter. In contrast, expectations based on the bootstrap particle filter or in the absence of any resampling are evidently unable to capture the correct behavior for the posterior variance. As a consequence, these particle filters will fail to provide appropriate credibility regions for the parameters and thereby make any reliable prediction impossible. This is also illustrated in table \ref{tab:1} where we compared the performances of our various resampling strategies, each time averaged over 10 independent runs. These results thus clearly demonstrate the ability of the RPF to outperform such traditional particle filters for a moderate number of particles.


\section{Conclusions}
\label{sec:Concl}

To conclude, we demonstrated that the choice of resampling strategy in the regularized particle filter can be paramount, as far as estimation of the posterior probability density is concerned. In particular, using analytical arguments inspired from linear Gaussian state-space models, we proved that resampling on a systematic basis causes the RMSE to be lower bounded by a constant which depends on the number of particles and on the dimensionality of the state space. Numerical evidence supporting our claims, including for non-linear state-space models, was also highlighted. Finally, we showed that if one introduces a resampling ratio, or if the kernel bandwidth is appropriately modulated, then the regularized particle filter clearly outperforms traditional bootstrap particle filters.

As a closing remark, we note that the idea we highlighted in this paper of progressively reducing the bandwidth perturbation seems related to recent approaches known as tempering \cite{Svensson17}, which consist in building a probability measure which depends on a tunable parameter allowing one to smoothly transport particles from the prior $\pi_0$ to the posterior $\pi_n$. Similar ideas have also been invoked to perform Bayesian update progressively \cite{Daum08}, and investigating the ramifications of such connections in more detail would certainly be worthy of interest. We also leave for future research the question of which statistical criterion (Kullback-Leibler divergence, Wasserstein metric, etc.) could justify this type of kernel bandwidth modulation.

\begin{table*}
\caption{Comparison of particle filter predictions for different resampling strategies at $n=100$ averaged over 10 runs.}
\label{tab:1}
\begin{tabular}{|l|l|l|l|l|l|}
\hline
 & Benchmark & RPF (ESS$_\text{crit}$ = 0.5) & RPF ($\alpha_n = (n+\alpha_h^{-1})^{-1}$) & SIR (ESS$_\text{crit}$ = 0.5) & SIS (ESS$_\text{crit}$ = 0) \\
\hline
$\mu_\text{RUE}$ & 3.521 & 3.479 (0.0096) & 3.505 (0.0291) & 3.554 (0.0431) & 3.503 (0.0569) \\
\hline
$\sigma_\text{RUE}$ & 0.0325 & 0.0487 (0.0041) & 0.0341 (0.0083) & 0.0126 (0.0161) & 0.0060 (0.0049) \\
\hline
$\mu_\gamma$ & 0.632 & 0.639 (0.0026) & 0.636 (0.0090) & 0.623 (0.0095) & 0.634 (0.0089) \\
\hline
$\sigma_\gamma$ & 0.0070 & 0.0146 (0.0016) & 0.0074 (0.0032) & 0.0024 (0.0030) & 0.0007 (0.0005) \\
\hline
$\mu_{\mu_a}$ & 542.6 & 537.1 (2.37) & 540.4 (7.55) & 547.2 (18.08) & 540.4 (10.50) \\
\hline
$\sigma_{\mu_a}$ & 8.297 & 14.43 (1.69) & 9.996 (1.74) & 4.764 (5.89) & 2.477 (3.09) \\
\hline
\end{tabular}
\end{table*}

\begin{acknowledgements}
The authors wish to thank Gautier Viaud for useful conversations at various stages of this work.
\end{acknowledgements}

\bibliography{RPF}

\begin{thebibliography}{10}
\providecommand{\url}[1]{{#1}}
\providecommand{\urlprefix}{URL }
\expandafter\ifx\csname urlstyle\endcsname\relax
  \providecommand{\doi}[1]{DOI~\discretionary{}{}{}#1}\else
  \providecommand{\doi}{DOI~\discretionary{}{}{}\begingroup
  \urlstyle{rm}\Url}\fi

\bibitem{Andrieu10}
Andrieu, C., Doucet, A., Holenstein, R.: Particle {M}arkov chain {M}onte
  {C}arlo methods.
\newblock J. Royal Stat. Soc. B \textbf{72}(3), 269--342 (2010)

\bibitem{Campillo09}
Campillo, F., Rossi, V.: Convolution particle filter for parameter estimation
  in general state-space models.
\newblock IEEE Transactions on Aerospace and Electronic Systems \textbf{45}(3),
  1063--1072 (2009)

\bibitem{Candes06}
Candes, E.J.: Modern statistical estimation via oracle inequalities.
\newblock Acta Numerica \textbf{15}, 257–325 (2006)

\bibitem{HMM}
Capp\'e, O., Moulines, E., Ryden, T.: Inference in Hidden Markov Models.
\newblock Springer Verlag (2005)

\bibitem{CarmierKalman}
Carmier, P.: Bayesian inference in linear {G}aussian hidden {M}arkov models: a
  tutorial on the {K}alman filter.
\newblock in preparation

\bibitem{Carpenter99}
Carpenter, J., Clifford, P., Fearnhead, P.: Improved particle filter for
  nonlinear problems.
\newblock IEE Proceedings - Radar, Sonar and Navigation \textbf{146}(1), 2--7
  (1999)

\bibitem{ChenPhD}
Chen, Y.: Bayesian inference in plant growth models for prediction and
  uncertainty assessment.
\newblock Ecole Centrale des Arts et Manufactures (2014)

\bibitem{Chopin02}
Chopin, N.: A sequential particle filter for static models.
\newblock Biometrika \textbf{89}, 539--552 (2002)

\bibitem{Chopin04}
Chopin, N.: Central limit theorem for sequential {M}onte {C}arlo methods and
  its application to {B}ayesian inference.
\newblock Annals of Statistics \textbf{32}, 2385--2411 (2004)

\bibitem{Chopin13}
Chopin, N., Jacob, P.E., Papaspiliopoulos, O.: {SMC}$^2$: an efficient
  algorithm for sequential analysis of state space models.
\newblock J. Royal Stat. Soc. B \textbf{75}(3), 397--426 (2013)

\bibitem{Cournede13}
Cournede, P.H., Chen, Y., Wu, Q., Baey, C., Bayol, B.: Development and
  evaluation of plant growth models: Methodology and implementation in the
  pygmalion platform.
\newblock Mathematical Modeling of Natural Phenomena \textbf{8}, 112--130
  (2013)

\bibitem{Daum08}
Daum, F., Huang, J.: Particle flow for nonlinear filters with log-homotopy.
\newblock Proc. SPIE \textbf{6969}, 696,918 (2008)

\bibitem{DelMoral04}
DelMoral, P.: Feynman-Kac Formulae: Genealogical and Interacting Particle
  Systems with Applications.
\newblock Springer (2004)

\bibitem{DelMoral06}
DelMoral, P., Doucet, A., Jasra, A.: Sequential {M}onte {C}arlo samplers.
\newblock J. Roy. Stat. Soc. B \textbf{68}(3), 411--436 (2006)

\bibitem{DelMoral12}
DelMoral, P., Doucet, A., Jasra, A.: On adaptive resampling strategies for
  sequential {M}onte {C}arlo methods.
\newblock Bernoulli \textbf{18}(1), 252--278 (2012)

\bibitem{SMC}
Doucet, A., de~Freitas, N., Gordon, N.: Sequential Monte Carlo methods in
  practice.
\newblock Springer Verlag (2001)

\bibitem{Gauchi13}
Gauchi, J.P., Villa, J.P.: Nonparametric particle filtering approaches for
  identification and inference in nonlinear state-space dynamic systems.
\newblock Statistics and Computing \textbf{23}, 523--533 (2013)

\bibitem{Gordon93}
Gordon, N., Salmond, D., Smith, A.: A novel approach to
  nonlinear/non-{G}aussian {B}ayesian state estimation.
\newblock IEE Proc. F \textbf{140}, 107--113(6) (1993)

\bibitem{Hol06}
Hol, J., Schon, T., Gustafsson, F.: On resampling algorithms for particle
  filters.
\newblock In: IEEE Nonlinear Statistical Signal Processing Workshop, pp. 79--82
  (2006)

\bibitem{Ionides11}
Ionides, E.L., Bhadra, A., Atchad\'e, Y., King, A.A.: Iterated filtering.
\newblock Ann. Stat. \textbf{39}(3), 1776--1802 (2011)

\bibitem{Jacob15}
Jacob, P.: Sequential {B}ayesian inference for implicit hidden {M}arkov models
  and current limitations.
\newblock ESAIM: Procs. \textbf{51}, 24--48 (2015)

\bibitem{Kalman60}
Kalman, R.: A new approach to linear filtering and prediction problems.
\newblock J. Basic Engineering \textbf{82}, 35--45 (1960)

\bibitem{Kalman61}
Kalman, R., Bucy, R.: New results in linear filtering and prediction theory.
\newblock J. Basic Engineering \textbf{83}, 95--107 (1961)

\bibitem{Kantas15}
Kantas, N., Doucet, A., Singh, S.S., Maciejowski, J.M., Chopin, N.: On particle
  methods for parameter estimation in state-space models.
\newblock Statistical Science \textbf{30}(3), 328--351 (2015)

\bibitem{Kong94}
Kong, A., Liu, J.S., Wong, W.H.: Sequential imputations and {B}ayesian missing
  data problems.
\newblock J. Americ. Stat. Assoc. \textbf{89}(425), 278--288 (1994)

\bibitem{LeGland98}
LeGland, F., Musso, C., Oujdane, N.: An analysis of regularized interacting
  particle methods for nonlinear filtering.
\newblock In: Proceedings of the 3rd IEEE European Workshop on
  Computer-Intensive Methods in Control and Signal Processing, pp. 167--174
  (1998)

\bibitem{Liu10}
Liu, J., West, M.: Combined parameter and state estimation in simulation-based
  filtering (Ch. 10), in Sequential Monte Carlo methods in practice.
\newblock Springer Verlag (2001)

\bibitem{Liu98}
Liu, J.S., Chen, R.: Sequential {M}onte {C}arlo methods for dynamic systems.
\newblock J. Americ. Stat. Assoc. \textbf{93}(443), 1032--1044 (1998)

\bibitem{May76}
May, R.: Simple mathematical models with very complicated dynamics.
\newblock Nature \textbf{261}, 459--467 (1976)

\bibitem{Musso12}
Musso, C., Oujdane, N., LeGland, F.: Improving regularized particle filters
  (Ch. 12), in Sequential Monte Carlo methods in practice.
\newblock Springer Verlag (2001)

\bibitem{Parzen62}
Parzen, E.: On estimation of a probability density function and mode.
\newblock Ann. Math. Stat. \textbf{33}(3), 1065--1076 (1962)

\bibitem{Paulin16}
Paulin, D., Jasra, A., Crisan, D., Beskos, A.: On concentration properties of
  partially observed chaotic systems.
\newblock unpublished  (2016).
\newblock \urlprefix\url{https://arxiv.org/abs/1608.08348}

\bibitem{Sheather04}
Sheather, S.: Density estimation.
\newblock Statistical Science \textbf{19}, 588--597 (2004)

\bibitem{Svensson17}
Svensson, A., Sch\"on, T.B., Lindsten, F.: Learning of state-space models with
  highly informative observations: a tempered sequential {M}onte {C}arlo
  solution.
\newblock unpublished  (2017).
\newblock \urlprefix\url{https://arxiv.org/abs/1702.01618}

\bibitem{Tsybakov}
Tsybakov, A.: Introduction to Nonparametric Estimation.
\newblock Springer (2009)

\bibitem{West93}
West, M.: Approximating posterior distributions by mixtures.
\newblock J. Roy. Stat. Soc. B \textbf{55}(2), 409--422 (1993)

\end{thebibliography}
\bibliographystyle{spmpsci}

\appendix

\section{Optimal bandwidth selection}
\label{app:band}

There are various ways \cite{Sheather04} of choosing the bandwidth parameter $h_n$ in Eq.~(\ref{eq:RPF}). The most popular choice relies on minimizing the mean integrated squared error (MISE) between the true posterior probability density function $p$ and the kernel density $p_h^N(x)=\sum_{i=1}^N {\cal K}_h(x-x^{(i)})/N$, with $x^{(i)} \sim p(x)$:
\be
\label{eq:MISE}
\begin{split}
\text{MISE}(h) & = \mathbb{E}\left[ ||p_h^N - p||_{L^2}^2 \right] 
\\
& = ||\mathbb{E}[p_h^N] - p||_{L^2}^2 + \int dz \; \mathbb{V}[p_h^N(z)] \; ,
\end{split}
\ee
where we introduced the $L^2(\mathbb{R})$ norm $p \mapsto ||p||_{L^2}=\sqrt{\int dz \; p(z)^2}$. As explicited above, the MISE can allegedly be decomposed into the integrated squared bias and the integrated variance. Observing that $\mathbb{E}[p_h^N(x)] = \int dz \; p(z) {\cal K}_h(x-z)$ and developing $p$ around $x$ as a Taylor series, one obtains:
\be
\mathbb{E}[p_h^N(x)] - p(x) = \frac{h^2}{2}p''(x)\int dz \; z^2 {\cal K}_1(z) + o(h^2) \; .
\ee
Likewise, the variance can be computed by decomposing $\mathbb{E}[p_h^N(x)^2]$ into diagonal and non-diagonal contributions:
\be
\begin{split}
\mathbb{E}[p_h^N(x)^2] & = \frac{1}{N} \int dz \; p(z) {\cal K}_h(x-z)^2 
\\
& + \frac{N(N-1)}{N^2}\left(\int dz \; p(z){\cal K}_h(x-z)\right)^2 
\end{split}
\ee
such that $\mathbb{V}[p_h^N(x)] = (||{\cal K}_1||_{L^2}^2 p(x) + o(1))/(Nh)$. Injecting these results in Eq.~(\ref{eq:MISE}) and minimizing with respect to $h$ is straightforward and provides the optimal choice:
\be
h_{\text{MISE}}^5 = \frac{||{\cal K}_1||_{L^2}^2}{N||p''||_{L^2}^2\left(\int dz \; z^2 {\cal K}_1(z)\right)^2} \; .
\ee
Note that this implies that $\text{MISE} \propto N^{-4/5}$. The only problem left is that this choice depends on the second moment of the underlying true probability density function which is generally unknown. There are different strategies to estimate this quantity, which we will not delve into \cite{Sheather04}. If we assume that both $p$ and ${\cal K}$ are Gaussian distributed, the expression for $h$ reduces to the simple $h_{\text{MISE}}^5 = 4\Sigma_p^{5/2}/(3N)$ which coincides with Eq.~(\ref{eq:bandwidth}) for $d_x=1$.

\section{Kalman filter equations}
\label{app:KF}

The Kalman filter equations (\ref{eq:Kalman}) can be derived recursively in two steps, starting from the probability density function at time $n-1$. The first step is the prediction step, in which we seek to determine $p(\hat{{\bf x}}_n | {\bf y}_{1:n-1})$, in other words how the probability density function is propagated between two time steps. Using Eq.~(\ref{eq:linHMM}), one can easily show that 
\be
\label{eq:predictx}
{\bf \mu}_{n|n-1} := {\mathbb E}[\hat{{\bf x}}_n | {\bf y}_{1:n-1}] = A_n {\bf \mu}_{n-1}
\ee
using that state noise has null expectancy, and
\be
\label{eq:predictSigma}
\begin{split}
\Sigma_{n|n-1} & := {\mathbb E}[(\hat{{\bf x}}_n - {\bf \mu}_{n|n-1})(\hat{{\bf x}}_n - {\bf \mu}_{n|n-1})^T] 
\\
& = A_n \Sigma_{n-1} A_n^T + Q_n
\end{split}
\ee
using that state noise is uncorrelated with the dynamics. The second step is the update step, in which we seek to compute the posterior moments ${\bf \mu}_n$ and $\Sigma_n$, expressions for which are traditionally obtained by enforcing unbiasedness of the estimator and minimizing the mean squared error, as we will now show. We begin by writing the posterior mean as a linear combination between the prior mean and the observation, that is ${\bf \mu}_n = K'_n {\bf \mu}_{n-1} + K_n {\bf y}_n$. Enforcing ${\mathbb E}[\hat{{\bf x}}_n - {\bf \mu}_n | {\bf y}_{1:n-1}]=0$ leads to $K'_n = (\mathbb{I} - K_nB_n)A_n$, and thus:
\be
\label{eq:postx}
{\bf \mu}_n = (\mathbb{I} - K_nB_n){\bf \mu}_{n|n-1} + K_n{\bf y}_n
\ee
which expresses the posterior mean as a linear combination between the predicted mean and the observation, with a tradeoff controlled by the so-called Kalman gain matrix $K_n$, an expression for which will be derived below. We begin by computing $\Sigma_n$, which follows directly from the previous equation:
\be
\label{eq:generalSigma}
\begin{split}
\Sigma_n & = \text{Cov}[(\mathbb{I} - K_nB_n)(\hat{{\bf x}}_n - {\bf \mu}_{n|n-1}) + K_n\eta_n] 
\\
& = (\mathbb{I} - K_nB_n)\Sigma_{n|n-1}(\mathbb{I} - K_nB_n)^T + K_nR_nK_n^T \; .
\end{split}
\ee
Here we defined $\text{Cov}[\hat{{\bf x}}] = \mathbb{E}[\hat{{\bf x}}\hat{{\bf x}}^T]$ as a shorthand and used the fact that observation noise is uncorrelated with the dynamics. The Kalman gain is then obtained by minimizing the mean squared error, which corresponds to the trace of the posterior covariance matrix: 
\be
\label{eq:MMSE}
\begin{split}
0 & = \frac{\partial \text{Tr}[\Sigma_n]}{\partial K_n} 
\\ 
& = -2\Sigma_{n|n-1}B_n^T + 2K_nB_n\Sigma_{n|n-1}B_n^T + 2K_nR_n \; ,
\end{split}
\ee
where derivation with respect to a matrix is performed element-wise. Solving Eq.~(\ref{eq:MMSE}) yields the Kalman gain expression
\be
\label{eq:Kgain}
K_n = \Sigma_{n|n-1}B_n^T(B_n\Sigma_{n|n-1}B_n^T + R_n)^{-1} \; ,
\ee
where matrix inversion is justified by positive definiteness. Injecting this in Eq.~(\ref{eq:generalSigma}), we obtain a much simpler expression for the posterior covariance matrix:
\be
\label{eq:Sigma}
\Sigma_n = (\mathbb{I} - K_nB_n)\Sigma_{n|n-1} \; .
\ee
From Eqs.~(\ref{eq:Kgain},\ref{eq:Sigma}), we also derive the useful and synthetic relation $\Sigma_nB_n^T=K_nR_n$. Finally, using this and rearranging Eqs.~(\ref{eq:postx},\ref{eq:Sigma}), we obtain the central result Eq.~(\ref{eq:Kalman}). Notice that we made no real use of the assumption of Gaussian distributions. In fact, the above results hold true irrespective of the underlying probability density function, provided noises are centered. The distinctive feature of the Gaussian case is that the knowledge of the first two moments allows recovering the entire distribution. Specifically, for Gaussian probability density functions, the Kalman filter equations can actually also be derived without any particular assumptions using Bayes rule.

\end{document}